\documentclass{article}
\pdfoutput=1
\usepackage{wright}
\usepackage{qcircuit}
\usepackage{ytableau}
\usepackage{subcaption}

\AtEveryBibitem{ % delete unnecessary information
    \clearfield{day}
    \clearfield{month}
    \clearfield{series}
    \clearfield{venue}
    \clearname{editor}
    \clearlist{publisher}
    \clearlist{location} % alias to field 'address'
    \clearfield{venue}
    \clearfield{issn}
    \clearfield{isbn}
    \clearfield{urldate}
    \clearfield{eventdate}
    \clearfield{pages}
    \clearfield{number}
    \clearfield{volume}
}

\DeclareMathOperator{\Shift}{Shift}

\title{Conjugate queries can help}
\author{
    Ewin Tang\thanks{UC Berkeley. \texttt{\{ewin,jswright\}@berkeley.edu}} 
    \and John Wright\footnotemark[1]
    \and Mark Zhandry\thanks{Stanford University \& NTT Research. \texttt{mzhandry@stanford.edu}}
}

\date{}

\begin{document}

\maketitle

\begin{abstract}
We give a natural problem over input quantum oracles $U$ which cannot be solved with exponentially many black-box queries to $U$ and $U^\dagger$, but which can be solved with constant many queries to $U$ and $U^*$, or $U$ and $U^{\mathrm{T}}$.
We also demonstrate a quantum commitment scheme that is secure against adversaries that query only $U$ and $U^\dagger$, but is insecure if the adversary can query $U^*$. These results show that conjugate and transpose queries do give more power to quantum algorithms, lending credence to the idea put forth by Zhandry that cryptographic primitives should prove security against these forms of queries.

Our key lemma is that any circuit using $q$ forward and inverse queries to a state preparation unitary for a state $\sigma$ can be simulated to $\varepsilon$ error with $n = \mathcal{O}(q^2/\varepsilon)$ copies of $\sigma$.
Consequently, for decision tasks, algorithms using (forward and inverse) state preparation queries only ever perform quadratically better than algorithms with sample access.
We also identify a motif, which we call the ``acorn trick'', where generically strengthening a quantum resource can be possible if the output is allowed to be random, bypassing no-go theorems for deterministic algorithms.
We demonstrate this idea for several settings, including controlization and purification.
\end{abstract}
%
%\newpage
\hypersetup{linktocpage}
\tableofcontents
%
%\newpage

\section{Introduction}

What does it mean to have access to a quantum procedure?
A routine on a quantum computer is typically specified by a sequence of gates from a specified base set of gates, which we think of as implementing the unitary matrix $U: \ket{\psi} \mapsto U\ket{\psi}$.
So, the most straightforward way to give you access to a routine $U$ is to give you a quantum circuit implementing $U$.

With this quantum circuit, you can implement the unitary $U$, which we call a \emph{forward query} to $U$.
But you can also do more: for example, you can also implement an \emph{inverse query}, $U^\dagger$, by inverting the (reversed) circuit gate-by-gate.
In a similar way, you can implement controlled queries, $\controlled U$, by replacing every gate with its controlled version.

Access to $U$, $U^\dagger$, and their controlled versions are standard in quantum algorithms.
These can be implemented efficiently when $U$ can, and so we often see algorithms using these alternative types of queries.
Interestingly, these queries enable procedures which cannot be done with forward queries alone.\footnote{
    For more detail, we refer the readers to prior work of a subset of the authors which investigate the power of inverse and controlled queries~\cite{tw25a,tw25b}.
}
So, we can ask: is there any more power we can squeeze out of a generic circuit for $U$?

This question was recently considered by Zhandry, where evidence was given that there are only two more kinds of oracle queries to consider: queries to the complex conjugate of $U$, $U^*$, and to its inverse, the transpose of $U$, $U^\trans$~\cite{zhandry25}.
A circuit for $U$ can be converted into a circuit for both $U^*$ and $U^\trans$ in the same gate-by-gate manner as for $U^\dagger$ and $\controlled U$.
However, these queries are rarely used in quantum algorithms, and thus are not well-understood.
We know that queries to $U^*$ and $U^\trans$ cannot be simulated from queries to $U$~\cite{ehmmqs23,qdssm19b}, so these queries do \emph{something}, but we would like to understand what they enable operationally.
What algorithms are we missing out on by excluding these queries from our space of quantum oracles?
Is there anything we can do with $U^*$ and $U^\trans$ that we could not do without them?

\subsection{Results}

\subsubsection{In algorithms}
We give a simple instance of a problem which requires conjugate queries or transpose queries to solve.
For simplicity, we state our results for conjugate queries, and note the modifications to handle transpose queries in \cref{subsec:trans}.
The problem is as follows.

\begin{problem}[Reality testing]
    \label{prob:real-test}
    For a pure state $\ket{\psi} \in \C^d$, consider the fidelity between it and its complex conjugate:
    \begin{align*}
        \real(\ket{\psi}) \coloneqq \abs{\braket{\psi^*}{\psi}}^2 = \abs[\Big]{\sum_{i=1}^d \psi_i^2}^2.
    \end{align*}
    Given access to $\ket{\psi}$, the goal of \emph{reality testing} is to decide whether $\real(\ket{\psi})$ is $1$ or smaller than $1/10$ with probability $\geq 2/3$.
\end{problem}
When $\ket{\psi}$ only has real amplitudes, then this quantity $\real(\ket{\psi})$ is $1$.
In fact, $\real(\ket{\psi}) = 1$ if and only if $e^{\ii \theta} \ket{\psi}$ has real amplitudes for some phase $e^{\ii \theta}$; further, we can relate $\real(\ket{\psi})$ to the fidelity of $\ket{\psi}$ to the closest state with real amplitudes (\cref{fact:closeness}).
So, we can think about $\real(\ket{\psi})$ as describing how close $\ket{\psi}$ is to being real, up to a global phase.

We have not yet described what kind of access we have to $\ket{\psi}$ in \cref{prob:real-test}.
If we are only given copies of $\ket{\psi}$, then this problem is hard.
In fact, the hardness of testing whether a quantum state has real amplitudes from states is well-known~\cite{chm21,hbcclmnbkpm22}.
However, we will typically not only have access to copies of a quantum state; we will also know the quantum circuit preparing the state, since often we are the ones preparing the states in the first place.
This circuit is known as a state preparation unitary.

\begin{definition}[State preparation unitary]
    \label{def:state-prep}
    For a state with density matrix $\rho$, we say that $U_{\reg{AB}} \in (\C^{\wh{d} \times \wh{d}})_{\reg{A}} \otimes (\C^{d \times d})_{\reg{B}}$ is a \emph{state preparation unitary} for $\rho$ if $\tr_\reg{A}(U (\proj{0}_{\reg{AB}}) U^\dagger) = \rho_{\reg{B}}$.
\end{definition}

We show that reality testing can be solved with a constant number of queries to $U$ and $U^*$, but cannot be solved with polynomially many queries to $U$ and $U^\dagger$.
This gives our desired separation between oracle access to forward and conjugate queries, versus access to just forward queries.

\begin{theorem}
    \label{thm:useful}
    Let $U \in \C^{2d^2 \times 2d^2}$ be an arbitrary state preparation unitary of $\proj{\psi}$.
    Then reality testing with $d > 1000$ requires $\bigOmega{d^{1/4}}$ queries if we are only given access to $U$ and $U^\dagger$.
    On the other hand, it can be solved with $2$ queries to $U$ and $U^*$.
\end{theorem}

In other words, reality testing is easy when we are given $\ket{\psi}$ via a polynomial-time state preparation circuit, since in this case both $U$ and $U^*$ can be performed efficiently.
(The algorithm is simple: prepare copies of $\ket{\psi}$ and $\ket{\psi^*}$, and then perform a swap test.)
Folk intuition might suggest that this task is hard: reality testing given copies of states is hard, and algorithms which use state preparation unitaries are typically only quadratically more efficient than algorithms which use copies of the state.
If we do not factor in our ability to make conjugate queries when given a state preparation oracle, we are led to incorrect beliefs about reality testing.

\begin{remark}[Separations using conjugate states]
    Prior work has proven a similar separation, except with conjugate states instead of conjugate queries.
    The papers of Haug, Bharti, and Koh~\cite[Corollary 1]{hbk25} and Somma, King, Kothari, O'Brien, and Babbush~\cite[Theorem 4]{skkob25} note that reality testing can be solved with $\bigO{1}$, copies of $\ket{\psi}$ and $\ket{\psi^*}$, but require $\bigOmega{\sqrt{d}}$ copies of $\ket{\psi}$.
    We note also the work of King, Wan, and McClean~\cite{kwm24} which investigates other kinds of advantages conjugate states bring in the regime where quantum space is limited.
    Our contribution is showing that an exponential lower bound still holds in the stronger input model where we are given forward and inverse access to a state preparation unitary for $\ket{\psi}$.
\end{remark}

Our main technical contribution formalizes the intuition that state preparation unitaries only help quadratically over copies of the state.

\begin{theorem}[Simulating queries to a state preparation unitary given copies of the state]
    \label{thm:useless}
    Consider a circuit which uses $q$ calls to $U$ and $U^\dagger$, where $U \in \C^{\wh{d} \times \wh{d}} \otimes \C^{d \times d}$ is a state preparation unitary of the mixed state $\sigma \in \C^{d \times d}$ with ancilla register size $\wh{d} \geq 2d$.
    Let $\rho(U)$ denote the output of this circuit when run on $U$.
    Then, using $n = \bigO{q^2 / \eps}$ copies of $\sigma$ and $\poly(q, \log(\wh{d}), 1/\eps)$ gates, we can simulate the output of this circuit on a distribution\footnote{
        This distribution is described more explicitly in \cref{sec:useless}.
        It is not the uniform distribution over state preparation unitaries; the distribution is supported entirely on reflections $V$.
    } over state preparation unitaries of $\sigma$, $\E_{V} \rho(V)$, to $\eps$ error in trace distance.
    If we are told in advance that $\sigma$ has rank $r$, then the ancilla register size can be reduced to $\wh{d} \geq 2r$.
\end{theorem}

\cref{thm:useless} shows that access to state preparation unitaries can only help quadratically for decision problems, if we only use forward and inverse queries to them.
Concretely, let $H_0$ and $H_1$ be disjoint classes of states; then this defines a decision problem: given access to a state $\sigma$, decide whether $\sigma \in H_0$ or $\sigma \in H_1$ with success probability $\geq 2/3$.
If $\sigma$ is in neither class, the algorithm can output anything.
Let $n_{\textup{s}}$ be the minimum number of copies of $\sigma$ needed to solve the decision problem, and let $n_{\textup{q}}$ be the minimum number of forward and inverse queries to a state preparation unitary of $\sigma$ needed to solve the decision problem.
Then by applying \cref{thm:useless} with $\eps \gets 0.1$ and using standard success amplification techniques, the circuit which solves the decision problem with $n_{\textup{q}}$ queries can be converted into a circuit which solves it with $\bigO{n_{\textup{q}}^2}$ copies of $\sigma$.
Together with the observation that state preparation unitaries can be used to generate copies of $\sigma$, we get that $n_{\textup{q}} \leq n_{\textup{s}} \leq \bigO{n_{\textup{q}}^2}$.
This gives the powerful consequence that sample complexity lower bounds lift to lower bounds against state preparation unitaries, though with the caveat that these lower bounds only hold against forward and inverse queries, and not conjugate queries (as \cref{thm:useful} shows).\footnote{
    A plausible conjecture is that sample complexity lower bounds against algorithms with copies of $\sigma \otimes \sigma^*$ lift to lower bounds against forward, inverse, conjugate, and transpose queries (along with their controlled variants~\cite{tw25b}) to state preparation unitaries.
    This does not follow from \cref{thm:useless}, as if we apply it twice, it gives conjugate and transpose queries to a \emph{different} state preparation unitary to the forward and inverse queries.
    This inconsistency may be fixable via the acorn trick.
}

Note that here, we need that the query circuit works for \emph{any} state preparation unitary of $U$.
This is a typical property for circuits which use state preparation unitaries, but it can be the case that a \emph{particular} choice of state preparation unitary may be more helpful than another.
An artificial way to engineer this is to embed the problem solution into the ancilla register of the state preparation unitary; a more realistic situation is when the ancilla register is size $\littleO{d}$, since then we can conclude that $\sigma$ is not full rank~\cite{lgdc24}.

Various weaker versions of \cref{thm:useless} are well-known, but as far as we know, this strongest form of the implication is new---in particular, the observation that this statement holds for mixed states as well as pure states.
See \cref{rmk:prior-useless-results} for more explanation.

\begin{remark}[Is \cref{thm:useless} tight?]
    There are two main complexity quantities in \cref{thm:useless}.
    First, there is $\bigO{n^2/\eps}$, the number of copies of $\sigma$ needed to simulate $n$ forward or inverse queries to a state preparation unitary for $\sigma$ to $\eps$ total error.
    This is at least tight in the $n$ parameter, because of the optimality of Grover's algorithm.
    Further, the main cost comes from QPCA, which is tight~\cite{klloy17}, suggesting that this may be tight.

    Second, there is the size of the ancilla register, $\wh{d}$, which can be as small as $2r$.
    We need that $\wh{d} \geq r$, since the maximally mixed state can only be purified with a register of dimension at least $r$.
    So, we are off by a factor of $2$.
    We did not try to improve this constant factor.
    However, we suspect that it is not necessary: this factor comes from \cref{lem:reference}, but here the ancilla is only used to find an explicit state which is guaranteed to be orthogonal to the initial state.
    This could also be done approximately by choosing a random explicit state, like a computational basis vector, or perhaps exactly using the ``acorn trick'', which we describe below.
\end{remark}

\subsubsection{In cryptography}
Turning to cryptography, we consider the case of quantum (bit) commitments. These allow a quantum sender to commit to a bit $b\in\{0,1\}$ by sending a quantum state $\rho_R$ to a receiver, while keeping a potentially entangled quantum state $\rho_S$ private. Commitments should be \emph{hiding}, meaning the receiver learns nothing about the bit $b$ from $\rho_R$. Later, the sender ``opens'' the commitment by revealing $b$ and $\rho_S$, at which point the receiver verifies $b$ against the joint system $\rho_{S,R}$, and either accepts or rejects. Commitments should be \emph{binding} in the sense that the sender, post-commitment, should be unable to change the bit $b$ and still cause the verifier to accept. Either hiding or binding in quantum commitments must be computational~\cite{Mayers97,LoChau98}, in the sense that either hiding only holds for computationally-bounded receivers, or binding only holds for computationally bounded senders.

Cryptographers often reason about cryptosystems in oracle models. Sometimes, this is to give a heuristic argument for security where otherwise a standard-model security proof is difficult or impossible. Other times, cryptographers use oracle models to prove black-box separations---namely to show that one cryptographic primitive cannot imply another primitive, under relativizing techniques (which capture most techniques in the field). In order to show that $A$ cannot be used to build $B$, an oracle is provided relative to which $A$ exists but $B$ does not. Proving such a result in particular requires proving the security of $A$ relative to the oracle. In an oracle model, a computationally bounded adversary can only make polynomially-many queries.

We demonstrate how oracle security proofs can go awry if the oracle is a unitary oracle allowing for only access to $U$ and $U^\dagger$, but not $U^*$ or $U^\trans$. Concretely, we show:
\begin{theorem}\label{thm:commitmentmain} Relative to a Haar random unitary $U$, there exists a commitment scheme that is statistically hiding and computationally binding against adversaries that can only query $U$ and $U^\dagger$, but is \emph{not} computationally binding if the adversary can query $U^*$ (or alternatively $U$ and $U^\trans$).
\end{theorem}
Our commitment scheme is simple: to commit to 0, the sender just sends one-half of each of $n$ EPR pairs, keeping the other half of each pair for itself. To commit to 1, the sender does the same, but applies $U$ to the receiver's halves. The scheme is hiding, since in either case $\rho_R$ is just the totally-mixed state.

We also show that the commitment scheme is (computationally) binding, if we restrict to efficient adversaries that can only query $U$ and $U^\dagger$. Proving this is the core technical difficulty in proving \cref{thm:commitmentmain}, and involves showing, roughly, that any algorithm which is able to implement $U^*$, even approximately, by making queries to $U,U^\dagger$ must make exponentially-many queries.

On the other hand, the scheme is not binding if the adversary can make even a single query to $U^*$. The adversary can apply $U^*$ to $\rho_S$ to change from a commitment to 0 to a commitment to 1. Such conjugate access would be possible in any ``real world'' commitment scheme, showing that the security proof relative to $U,U^\dagger$ does not reflect real-world security.

\subsection{The acorn trick}

Our main technique is what we call the \emph{acorn trick}.
Let's start by seeing the trick on a simple example.

Suppose we have $n$ copies of a pure state $\ket{\psi}$, but we want copies of $\ket{\psi(\theta)} \coloneqq \frac{1}{\sqrt{2}}(e^{\ii \theta}\ket{0}\ket{0} + \ket{1}\ket{\psi})$ instead.
One reason we might want to do this is for tomography, where the latter state is called a ``conditional sample''~\cite{vAcgn22}.
Given only access to $\ket{\psi}$, if we want to estimate some of its entries, naive algorithms will only give us access to entry magnitudes, $\abs{\braket{i}{\psi}}$.
This is to be expected, since observables of $\ket{\psi}$ are invariant under global phase, but if every amplitude is normalized, we cannot even get a handle on relative phases between entries.
Conditional samples allow for easy estimation of $e^{-\ii \theta}\braket{i}{\psi}$ efficiently, where the $e^{-\ii \theta}$ is consistent across entries, fixing the relative phase problem across estimates.

We might then hope to design a circuit to deterministically convert any state $\ket{\psi}$ into a conditional version $\ket{\psi(\theta)}$.
However, there are simple barriers to constructing such a ``universal conditioner''.
This operation is unphysical, as it maps the undetectable global phase of $\ket{\psi}$ to a detectable relative phase of a conditional state.
Concretely, there is no way to map every $\ket{\psi}$ to a corresponding conditional state in a way which is consistent with linearity of a unitary circuit.
So, we instead ask for a protocol with a randomized output.

It is trivial to convert one copy of $\ket{\psi}$ to one random conditional sample, $\ket{\psi(\theta)}$, where $\theta$ is a uniformly random phase.
This is because the random phase decouples the $\ket{0}\ket{0}$ and $\ket{1}\ket{\psi}$ parts of the superposition, and so the random $\ket{\psi(\theta)}$ is simply either $\ket{0}\ket{0}$ or $\ket{1}\ket{\psi}$, each with $1/2$ probability.
However, this does not suffice for us: if we convert every sample into a conditional sample in this manner, then the phases $\theta$ will not be consistent across samples, ruining the phases on our estimates.
Nevertheless, it turns out that there is a way to convert $n$ copies of $\ket{\psi}$ to $n$ copies of $\ket{\psi(\theta)}$, where the random $\theta$ is consistent across all the conditional samples (\cref{lem:reference}).

This is what we call the acorn trick.
The general form of the trick is as follows: sometimes, we have a weak form of copies of a resource in quantum information, where we would instead rather have a stronger, ``lifted'' form of it, which may furnish greater quantum control, or enable the use of stronger subroutines.
The acorn trick applies in the setting where (1) there are many valid lifts of a resource, and (2) while it may be intractable to \emph{deterministically} map a resource to a lift, it is possible to map one to a \emph{random} lift.
Then, the trick states that it is possible to convert \emph{all copies} of that resource to a random lift of the resource, such that the lift is \emph{consistent across all copies}.\footnote{
    We choose this name with analogy to the behavior of oak trees.
    Some oaks have ``mast years'' where they produce many more acorns than normal.
    These mast years occur on a semi-regular cycle, and \emph{the cycle's phase is consistent across copies of the tree in a region}, e.g.\ all the oaks in a forest will mast on the same year~\cite{pkk16}.
}
In a sense, the acorn trick is a reframing of a symmetrization argument: objects with symmetry can be reduced to lower-dimensional objects without symmetry; by the same argument, we can take objects without symmetry and, with an algorithm which is agnostic to the specific identity of the object, lift them to a higher dimension, provided we do so with sufficient symmetry.

We have seen this trick play a key role in our investigations into analyzing unitary oracles and modeling quantum computational problems more broadly.
We use the acorn trick to:
\begin{enumerate}
    \item Convert samples into conditional samples~(\cite{gz25}, \cref{lem:reference});
    \item Convert unitaries into controlled unitaries~(\cite{tw25b});
    \item Convert mixed states into purifications of mixed states~(\cref{lem:purify}).\footnote{
        The independent and concurrent work of Ananth and Goldin~\cite{ag25} gives a different algorithm to perform this task. 
        We compare our result with theirs in \cref{sec:useless}.
    }
\end{enumerate}
These consequences can be surprising: the second and third results listed above bypass no-go theorems which state that universal controlization and purification are unphysical and, therefore, impossible~\cite{afcb14,kkmb06,dp13,fl20,ldcl25}.
With the acorn trick, we can perform both if we are permitted to output a random controlization and purification, respectively.
We anticipate more interesting applications of the acorn trick.

\paragraph{Acorn tricks for channels.}
One test of the acorn trick is the following.

\begin{conjecture}[Channel dilations don't help] \label{conj:dilation}
    For a channel $\Phi$ over $d$-dimensional states, call the unitary $U \in \C^{\wh{d} \times \wh{d}}_{\reg{A}} \otimes \C_{\reg{B}}^{d \times d}$ a \emph{dilation} of $\Phi$ if $\tr_{\reg{A}}(U (\proj{0}_\reg{A} \otimes \rho_{\reg{B}}) U^\dagger) = \Phi(\rho)_{\reg{B}}$.
    
    Consider a circuit which applies $U$, a dilation of $\Phi$ with $\wh{d} \gg d$, $q$ times and outputs the state $\rho(U)$.
    Then this can be converted into a circuit which applies $\Phi$ $\poly(q)$ times and outputs $\E_V \rho(V)$, where the expectation is over a random distribution of dilations $V$.
\end{conjecture}

This is a channel version of \cref{lem:purify}: instead of purifying copies of mixed states, we ask to convert applications of $\Phi$ to applications of a random Stinespring dilation of $\Phi$.\footnote{
    We note that, in general it is not possible to use $q$ queries to $\Phi$ to simulate $q$ queries to a dilation $U$, so the polynomial loss in \cref{conj:dilation} is inherent.
    To see this, consider the qubit channel $\Phi$ which discards its input and outputs a $p$-biased coin $p\cdot \proj{0} + (1-p) \cdot \proj{1}$. Then $\Omega(1/\epsilon^2)$ applications of $\Phi$ are required to estimate $p$ up to accuracy $\epsilon$. A dilation of $\Phi$ will be any unitary $U$ of dimension, say, $\C_{\reg{A}}^{2 \times 2} \otimes \C_{\reg{B}}^{2 \times 2}$
    such that 
    \begin{equation*}
        U \cdot \ket{0}_{\reg{A}}\ket{0}_{\reg{B}} = \sqrt{p} \cdot \ket{0}_{\reg{A}}\ket{u_0}_{\reg{B}} + \sqrt{1-p}\cdot \ket{1}_{\reg{A}}\ket{u_1}_{\reg{B}},
        \text{ and }
        U \cdot \ket{0}_{\reg{A}}\ket{1}_{\reg{B}} = \sqrt{p} \cdot \ket{0}_{\reg{A}}\ket{u_1}_{\reg{B}} + \sqrt{1-p} \cdot \ket{0}_{\reg{A}}\ket{u_0}_{\reg{B}},
    \end{equation*}
    where $\ket{u_0}$ and $\ket{u_1}$ are orthogonal vectors.
    But using amplitude estimation, we can output a number $\widehat{a}$ such that $\sqrt{p} - \epsilon < \widehat{a} < \sqrt{p} + \epsilon$, and setting $\widehat{p} = \widehat{a}^2$, we see that $\widehat{p}$ is $O(\epsilon)$-close to $p$.
    As the dimension of $U$ is $4$, amplitude estimation can be solved in $O(1/\epsilon)$ queries to $U$.
    This means that there must be at least a quadratic loss when simulating queries to a random $U$ with queries to the channel $\Phi$.
}

Subsequent to the publication of the arXiv version of this paper, several works have shown that certain variants of this conjecture are in fact true.
Because of these results, we realize now that our \cref{conj:dilation} is actually a composition of two conjectures.
The first states that a random \emph{dilation isometry} of a channel can be simulated with only queries to the channel.
The second states that a random \emph{completion} of an isometry can be simulated with only queries to the isometry.
A dilation, in the sense described in \Cref{conj:dilation}, is a completion of a dilation isometry, so that conjecture follows from the two individual conjectures.

\begin{conjecture}[Channel dilation isometries don't help] \label{conj:dilation-isometry}
    For a channel $\Phi$ over $d$-dimensional states, call the isometry $V: \C^d \to \C_{\reg{A}}^{\wh{d}} \otimes \C_{\reg{B}}^d $ a \emph{dilation (isometry)} of $\Phi$ if $\tr_{\reg{A}}(V \rho V^\dagger) = \Phi(\rho)_{\reg{B}}$.

    Consider a circuit which applies $V$, a dilation isometry of $\Phi$ with $\wh{d} \gg d$, $q$ times and outputs the state $\rho(U)$.
    Then this can be converted into a circuit which applies $\Phi$ $q$ times and outputs $\E_W \rho(W)$, where the expectation is over a random distribution of dilation isometries $W$.
\end{conjecture}

\begin{conjecture}[Isometry completions don't help] \label{conj:completion}
    For an isometry $V: \C_{\reg{B}}^d \to \C_{\reg{A}}^{\wh{d}} \otimes \C_{\reg{B}}^d$, call $U: \C_{\reg{A}}^{\wh{d}} \otimes \C_{\reg{B}}^d \to \C_{\reg{A}}^{\wh{d}} \otimes \C_{\reg{B}}^d$ a \emph{completion} of $V$ if $U\ket{0}_{\reg{A}} = V$.

    Consider a circuit which applies $U$, a completion of $V$, $q$ times and outputs the state $\rho(U)$.
    Then this can be converted into a circuit which applies $V$ $\poly(q)$ times and outputs $\E_W \rho(W)$, where the expectation is over a random distribution of completions $W$.
\end{conjecture}

Progress has been made on \cref{conj:dilation-isometry}.
\cite{gmzfl25,ynm25} show that we can use $q$ nonadaptive queries to a channel $\Phi$ to simulate $q$ nonadaptive queries to a uniformly random dilation isometry $V$ of $\Phi$; here, ``nonadaptive'' means that the $q$ applications of $V$ are made in parallel, so that the circuit is applying $V^{\otimes q}$.
We believe that the restriction to nonadaptive queries is an artifact of their proof, and thus we conjecture in \cref{conj:dilation-isometry} that $q$ queries to $\Phi$ are sufficient to simulate $q$ completely general (and therefore possibly adaptive) queries to a random dilation isometry of $\Phi$.

As for \cref{conj:completion}, we prove in \cref{thm:useless} the special case where $d = 1$.
The general case is open.

\paragraph{Acorn tricks for continuous-time evolution?}
We can also ask if converting between resources is possible even when the acorn trick does not apply.
A natural extension to our work is to consider whether ``continuous'' access, access to $U^t$ for every $t \in \R$, can ever help over ``discrete'' access, access to only $U$ and $U^\dagger$.
A simple case of this question is to ask whether access to root queries, $\sqrt{U}$, ever helps to solve a problem.
Again, along the lines of this and prior work, the answer seems to be yes and no.
A \emph{particular} root query can help: distinguishing between $U = [\begin{smallmatrix} e^{\ii (\pi - \eps)} \\ & e^{\ii (\pi - \eps)} \end{smallmatrix}]$ and $U = [\begin{smallmatrix} e^{\ii (\pi - \eps)} \\ & e^{- \ii (\pi - \eps)} \end{smallmatrix}]$ requires $\poly(1/\eps)$ queries to $U$ and $U^\dagger$, since the channels are close in diamond distance.
However, one root query suffices to distinguish the two, provided that $\sqrt{U}$ is the square root which maps the eigenvalue $e^{\ii \theta}$ to $e^{\ii \theta / 2}$ for $\theta \in (-\pi, \pi)$.
However, there are many possible choices of $\sqrt{U}$: can root access help, even when we do not get to choose which root?

Simulating root queries from black box queries to $U$ and $U^\dagger$ has been studied previously, but is only possible in limited settings~\cite{smm09,bccks17,gslw18}.
For example, it is not known how to simulate root queries efficiently for a Haar-random $U$.
The acorn trick also does not apply here, since it is not even clear how to simulate one random root query.
For a matrix $U$ with distinct eigenvalues, the map which takes $\rho$ to $\sqrt{U} \rho \sqrt{U}^\dagger$ for a random choice of square root is a depolarizing channel in the eigenbasis of $U$.
It's not clear how to simulate this using queries to $U$ and $U^\dagger$.
Therefore, we posit that, if the acorn trick does not apply, then root queries might in fact be helpful.

\begin{conjecture}[Root queries can help]
    There is a problem on oracles $U \in \C^{d \times d}$ whose query complexity given access to $U$ and $U^\dagger$ is exponentially larger than its complexity when given access to $U$, $U^\dagger$, and \emph{any} $V$ such that $V^2 = U$.
\end{conjecture}

\section{Simulating state preparation unitaries with copies of the state} \label{sec:useless}

In this section, we prove \cref{thm:useless}, that copies of a state $\sigma$ can be used to simulate forward and inverse queries to a state preparation unitary for $\sigma$.
This result follows from a straightforward combination of known results.
We could not locate this precise consequence in the literature, so we describe it here.

Our proof combines three lemmas.
First, we show how to convert copies of a pure state $\ket{\psi}$ to conditional copies of that pure state, $\frac{1}{\sqrt{2}}(e^{\ii \theta}\ket{0}\ket{0} + \ket{1}\ket{\psi})$ (\cref{lem:reference}).
Here, $\theta$ is a uniformly random phase; this is necessary because we are converting $\ket{\psi}$, which is invariant under global phase, to a conditional state, which is not invariant under $\ket{\psi}$'s global phase.
This lemma has appeared in works by Goldin and Zhandry~\cite[Lemma 4.4]{gz25} and Kretschmer \cite[Lemma 25]{kretschmer21}.

Second, we show how to convert copies of a pure state $\ket{\psi}$ to a reflection about the state $2\proj{\psi} - I$ (\cref{lem:qpca}).
We do this with quantum principal component analysis~\cite{lmr14,klloy17}, commonly seen in the algorithms literature.
In the complexity and cryptography literature, a different algorithm for this task more common: reflect $\ket{\psi} \dots \ket{\psi} \ket{\varphi}$ about the symmetric subspace, and then trace out the first $n$ registers \cite[Lemma 42]{aru14}; \cite[Theorem 5]{jls18}; \cite[Proposition 5.1]{gz25}.
This is arguably cleaner, but is suboptimal, requiring a quadratically worse $\bigO{1/\delta^2}$ copies to implement a channel which is $\delta$-close to a reflection.

The first two lemmas are enough to prove \cref{thm:useless} for pure states.
To get it for mixed states, we show how to convert copies of a mixed state to copies of a random purification of the state (\cref{lem:purify}).
This parallels results of Soleimanifar and Wright~\cite{sw22} and Chen, Wang, and Zhang~\cite{cwz24}; the main difference is that the perspective of these works is different.
These works do not give a purification algorithm.
Instead, they prove lower bounds, showing that we can remove the dependence on purifications from algorithms for certain kinds of problems.
We prove an upper bound, so we identify precisely why we may not need purifications: there is an algorithm which produces purifications without any loss.
This change in perspective strengthens this claim and highlights its counter-intuitiveness.
We typically imagine a purification of a mixed state to be a stronger resource than the mixed state itself.
This results shows that purifications are only useful if they come with structure: a random purification is not helpful.

A similar result on random purification was also proved independently and concurrently by Ananth and Goldin~\cite{ag25}.
Their algorithm requires no representation theory to understand, but incurs an error of $n^2 / d$, where $n$ is the number of samples to purify and $d$ is the dimension of the sample; our algorithm is exact, and can be implemented to arbitrarily high precision on a quantum computer.

\begin{remark}[Comparison of \cref{thm:useless} to prior results] \label{rmk:prior-useless-results}
There are several similar results which appear prior in the literature.
Notably, \cref{thm:useless} for \emph{pure states} has appeared in the work of Goldin and Zhandry~\cite{gz25}, proceeding along the same lines as we do, and getting a complexity of $\bigO{q^2 / \eps^2}$ samples to simulate $q$ queries.
They use these results to understand the cryptographic model where parties are given a common Haar-random state; we recast it as a more general statement about state preparation unitaries.

In the quantum algorithms literature, it is well known that copies of a state $\sigma$ can be used to simulate forward and inverse queries to a block-encoding of $\sigma$, using the language of the QSVT framework~\cite{gslw18}.
This has been used to relate the strength of the block-encoding model to the state model~\cite{gp22,wz23,wz25}.
However, having a block-encoding to a state $\sigma$ is significantly weaker than having access to a state preparation unitary of $\sigma$: it's not even clear how to use block-encodings of $\sigma$ to prepare $\sigma$, even when it is pure.
It is surprising to us that this relatively weak statement can be strengthened so greatly with the two lemmas applying the acorn trick.
\end{remark}

\subsection{Converting samples into conditional samples}

\begin{definition}[Conditional samples with phase]
Given a state $\ket{\psi} \in \C^d$ and an angle $\theta \in [0, 2\pi)$, we define the corresponding \emph{conditional sample} as
\begin{equation*}
    \ket{\psi(\theta)} \coloneqq \frac{1}{\sqrt{2}}(e^{\ii \theta} \ket{0} \ket{0} + \ket{1} \ket{\psi}).
\end{equation*}
\end{definition}

The main result of this subsection is that given $n$ copies of a pure state $\ket{\psi}$, we can efficiently generate $n$ copies of a conditional sample $\ket{\psi(\theta)}$, where $\theta$ is a uniformly random angle.
As we mentioned before, this has been proven previously~\cite{gz25}; we prove it in a similar manner to prior ``acorn trick'' results~\cite{tw25b}.

\begin{lemma}[{Adding a reference phase to a pure state \includegraphics[height=2.0ex]{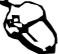}}]
    \label{lem:reference}
    There is a unitary circuit $\circuit^{(n)}$ such that, for all $\ket{\psi} \in \C^d$,
    \begin{align}\label{eq:goal}
        \circuit(\proj{\psi}^{\otimes n}) = \E_{\theta \sim [0, 2\pi)} \ketbra{\psi(\theta)}{\psi(\theta)}^{\otimes n}.
    \end{align}
    Moreover, this circuit is efficient and has gate complexity $\bigO{n^2 \log(d)}$.
\end{lemma}

We will begin by writing down a convenient expression for the mixed state on the right-hand side of \Cref{eq:goal} in terms of the following family of states.
\begin{definition}
    For shorthand, write $\ket{\psi^0} = \ket{0}\ket{0}$ and $\ket{\psi^1} = \ket{1}\ket{\psi}$.
    More generally, for $x \in \{0, 1\}^n$, we write
    \begin{equation*}
        \ket{\psi^x} = \ket{\psi^{x_1}}  \ket{\psi^{x_2}} \cdots \ket{\psi^{x_n}}.
    \end{equation*}
    Then for $0 \leq k \leq n$, define
    \begin{align*}
        \ket{\psi_n(k)} \coloneqq \sum_{x \in \{0, 1\}^n, |x| = k} \ket{\psi^x}.
    \end{align*}
    We will often write this simply as $\ket{\psi(k)}$ when $n$ is clear from context.
\end{definition}

Using these states, we can write an element of our desired mixture as follows.

\begin{lemma}
    \label{lem:conditional-sample-defn}
    Given an angle $\theta \in [0, 2\pi)$,
    \begin{equation*}
        \ket{\psi(\theta)}^{\otimes n}
        =
        \frac{1}{\sqrt{2^n}} \sum_{k = 0}^n e^{\ii \theta \cdot (n - k)} \ket{\psi(k)}.
    \end{equation*}
\end{lemma}
\begin{proof}
    This follows because
    \begin{align*}
        \Big(\frac{1}{\sqrt{2}} (e^{\ii \theta}\ket{0}\ket{0} + \ket{1}\ket{\psi})\Big)^{\otimes n}
        &= \frac{1}{\sqrt{2^n}} \sum_{x \in \{0, 1\}^n} (e^{\ii \theta \cdot (1 - x_1)}\cdot\ket{\psi^{x_1}}) \cdots (e^{\ii \theta \cdot (1 - x_n)}\cdot \ket{\psi^{x_n}}) \\
        &= \frac{1}{\sqrt{2^n}} \sum_{x \in \{0, 1\}^n} e^{\ii \theta \cdot (n - |x|)} \ket{\psi^x}
        = \frac{1}{\sqrt{2^n}} \sum_{k = 0}^n e^{\ii \theta \cdot (n - k)} \ket{\psi(k)}.\qedhere
    \end{align*}
\end{proof}

Now we average over $\theta$ to derive our expression for the right-hand side of \Cref{eq:goal}.

\begin{lemma}[Average conditional sample]
    \begin{equation*}
    \E_{\theta \sim [0, 2\pi)} \ketbra{\psi(\theta)}{\psi(\theta)}^{\otimes n}
    = \frac{1}{2^n} \sum_{k=0}^n\ketbra{\psi(k)}{\psi(k)}
    \end{equation*}
\end{lemma}

\begin{proof}
    Averaging over the random angle $\theta$, using \cref{lem:conditional-sample-defn},
    \begin{align*}
    \E_{\theta \sim [0, 2\pi)} \ketbra{\psi(\theta)}{\psi(\theta)}^{\otimes n}
    &{} = \E_{\theta \sim [0, 2\pi)} \Big(\frac{1}{\sqrt{2^n}} \sum_{k = 0}^n e^{\ii \theta \cdot (n - k)} \ket{\psi(k)}\Big) \cdot \Big(\frac{1}{\sqrt{2^n}} \sum_{\ell = 0}^n e^{-\ii \theta \cdot (n - \ell)} \bra{\psi(\ell)}\Big)\\
    &{} = \frac{1}{2^n} \sum_{k, \ell=0}^n \E_{\theta \sim [0, 2\pi)} [e^{\ii \theta \cdot (\ell - k)}] \cdot \ketbra{\psi(k)}{\psi(\ell)}\\
    &{} = \frac{1}{2^n} \sum_{k=0}^n\ketbra{\psi(k)}{\psi(k)},
    \end{align*}
    where in the last step we used the fact that $\E_{\theta \sim [0, 2\pi)} [e^{\ii \theta \cdot (\ell - k)}] = 1$ if $k = \ell$ and $0$ otherwise.    
\end{proof}

    It is this form of the mixture that we aim to simulate with $\circuit(\proj{\psi}^{\otimes n})$.
    The circuit $\circuit$ will consist of four registers named $\reg{C}$, $\reg{O}$, $\reg{S}$ and $\reg{K}$,
    which are structured as follows.
    \begin{itemize}
        \item The ``control'' register $\reg{C}$ contains $n$ qubits. These correspond to the $n$ control qubits in the state $\ket{\psi(k)}$.
        \item The ``output'' register $\reg{O}$ contains $n$ qudits in $\C^d$. These correspond to the $n$ qudits in the state $\ket{\psi(k)}$ which are either set to $\ket{0}$ or $\ket{\psi}$.
        \item The ``state'' register $\reg{S}$ contains $n$ qudits which are initialized to the $n$ copies of the input state $\ket{\psi}$.
        \item The ``clock'' register $\reg{K}$ contains a $\lceil \log_2(n+1)\rceil$ qubit counter. It keeps track of the number of $\ket{\psi}$'s which are currently in the output register.
    \end{itemize}
    The circuit $\circuit$ is depicted in \Cref{fig:circ-reference-phase}.
    It makes use of the $\Shift$ gate, defined as follows.
    \begin{definition}[The Shift gate]\label{def:shiftgate}
        The $\Shift$ gate acts on $(\C^d)^{\otimes n}$ and is defined as
    \begin{equation*}
        \Shift = \mathrm{SWAP}^d_{n-1, n} \cdots \mathrm{SWAP}^d_{2, 3} \cdot \mathrm{SWAP}^d_{1, 2},
    \end{equation*}
    where $\mathrm{SWAP}^d$ is the swap gate on $\C^d \otimes \C^d$.
    The  $\Shift$ gate has the property that for any qudit states $\ket{a_1}, \ldots, \ket{a_n} \in \C^d$,
    \begin{equation*}
        \Shift \cdot \ket{a_1, a_2, \ldots, a_{n-1}, a_n}
        = \ket{a_2, a_3, \ldots, a_n, a_1}.
    \end{equation*}
    \end{definition}
    The circuit $\circuit$ applies a sub-circuit known as the $\psi$-gadget $n$ times, one for each control register/output register pair.
    For intuition, let us consider the first application of the $\psi$-gadget, which acts as follows.
    \begin{enumerate}
    \item 
    It begins by applying a Hadamard gate to the first qubit, mapping it to $\frac{1}{\sqrt{2}} \ket{0} + \frac{1}{\sqrt{2}} \ket{1}$.
    \item
    It then checks if the first control bit is 1.
    If it is not, then the $\psi$-gadget does nothing more.
    \item 
    If it is, then the $\psi$-gadget swaps a copy of $\ket{\psi}$ into the first part of the output register $\reg{O}$,
    In doing so, it puts the $\ket{0}$ that used to be in that register into the first part of the state register $\reg{S}$.
    so that this register becomes $\ket{0} \otimes \ket{\psi} \otimes \cdots \otimes \ket{\psi}$.
    \item
    Next, it applies a $\Shift$ gate to the state register $\reg{S}$,
    so that it becomes $\ket{\psi} \otimes \cdots \otimes \ket{\psi} \otimes \ket{0}$.
    The point of this operation is to place a new $\ket{\psi}$ into the first part of the $\reg{S}$ register, so that it can be used in the next $\psi$-gadget.
    \item Finally, it increments the counter $\reg{K}$ from $\ket{0}$ to $\ket{1}$.
    \end{enumerate}
    In the $i$-th step, the $\psi$-gadget checks if the $i$-th control bit is 1 and, if so, swaps a copy of $\psi$ into the $i$-th part of the output register $\reg{O}$.

        \begin{figure}[ht]
\[
    \Qcircuit @R=1em @C=1em {
    \lstick{\ket{0}_{\reg{C},1}}
    & \gate{H}
    & \ctrl{6}
    & \ctrl{6}
    & \ctrl{9}
    & \qw
    & {\cdots}
    &
    & \qw
    & \qw
    & \qw
    & \qw
    & \qw\\
    \lstick{\makebox[\widthof{$\ket{0}$}]{\raisebox{0.6ex}{$\vdots$}}}
    & \qw
    & \qw
    & \qw
    & \qw
    & \qw
    & {\cdots}
    &
    & \qw
    & \qw
    & \qw
    & \qw
    & \qw\\
    \lstick{\ket{0}_{\reg{C},n}}
    & \qw
    & \qw 
    & \qw
    & \qw
    & \qw
    & {\cdots}
    &
    & \gate{H}
    & \ctrl{4}
    & \ctrl{4}
    & \ctrl{7}
    & \qw\\
    \lstick{\ket{0}_{\reg{O},1}}
    & \qw
    & \qswap
    & \qw
    & \qw
    & \qw
    & {\cdots}
    &
    & \qw
    & \qw
    & \qw
    & \qw
    & \qw \\
    \lstick{\makebox[\widthof{$\ket{0}$}]{\raisebox{0.6ex}{$\vdots$}}}
    & \qw
    & \qw
    & \qw
    & \qw
    & \qw
    & {\cdots}
    &
    & \qw
    & \qw
    & \qw
    & \qw
    & \qw\\
    \lstick{\ket{0}_{\reg{O},n}}
    & \qw
    & \qw
    & \qw
    & \qw
    & \qw
    & {\cdots}
    &
    & \qw
    & \qswap
    & \qw
    & \qw
    & \qw\\
    \lstick{\ket{\psi}_{\reg{S},1}}
    & \qw
    & \qswap
    & \multigate{2}{\Shift}
    & \qw
    & \qw
    & {\cdots}
    &
    & \qw
    & \qswap
    & \multigate{2}{\Shift}
    & \qw
    & \qw\\
    \lstick{\makebox[\widthof{$\ket{0}$}]{\raisebox{0.6ex}{$\vdots$}}}
    & \qw
    & \qw
    & \ghost{\Shift}
    & \qw
    & \qw
    & {\cdots}
    &
    & \qw
    & \qw
    & \ghost{\Shift}
    & \qw
    & \qw\\
    \lstick{\ket{\psi}_{\reg{S},n}}
    & \qw
    & \qw
    & \ghost{\Shift}
    & \qw
    & \qw
    & {\cdots}
    &
    & \qw
    & \qw
    & \ghost{\Shift}
    & \qw
    & \qw\\
    \lstick{\ket{0}_{\reg{K}}}
    & \qw
    & \qw
    & \qw
    & \gate{\Add(+1)}
    & \qw
    & {\cdots}
    &
    & \qw
    & \qw
    & \qw
    & \gate{\Add(+1)}
    & \qw \gategroup{1}{2}{10}{5}{.7em}{--} \gategroup{3}{9}{10}{12}{.7em}{--}
}
\]
\caption{
    The circuit $\circuit$.
    The sub-circuit in the dotted box is known as the \emph{$\psi$-gadget}.
}
\label{fig:circ-reference-phase}
\end{figure}
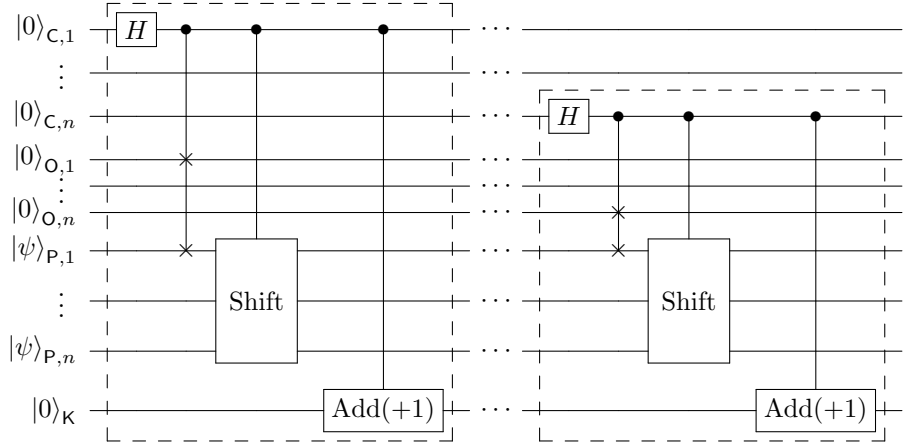

In general, the output of the circuit is given as follows.

\begin{lemma}
    On input $\ket{\psi}^{\otimes n}$, the output of the circuit is
    \begin{equation*}
        \ket{\circuit(\psi)}_{\reg{COSK}} =
        \frac{1}{\sqrt{2^n}} \sum_{k=0}^n \ket{\psi(k)}_{\reg{C}\reg{O}} \otimes (\ket{\psi}^{\otimes n - k} \otimes \ket{0}^{\otimes k})_{\reg{S}} \otimes \ket{k}_{\reg{K}}.
    \end{equation*}
\end{lemma}
\begin{proof}
    We will prove the more general statement that after the $i$-th $\psi$-gadget is applied, the state of the circuit looks like
    \begin{equation}\label{eq:ith-step}
        \frac{1}{\sqrt{2^i}} \sum_{k=0}^i \ket{\psi_i(k)}_{\reg{CO}_{[1, i]}} \otimes \ket{0}_{\reg{CO}_{[i+1,n]}} \otimes (\ket{\psi}^{\otimes n -k} \otimes \ket{0}^{\otimes k})_{\reg{S}} \otimes \ket{k}_{\reg{K}}.
    \end{equation}
    This implies the lemma statement because when $i = n$, $\ket{\psi_n(k)} = \ket{\psi(k)}$ by definition. We will prove this statement by induction. The $i = 0$ base case is true because the state is initialized to $\ket{0}_{\reg{C}} \ket{0}_{\reg{O}} \ket{\psi}^{\otimes n}_{reg{S}} \ket{0}_{\reg{K}}$ at the beginning of the circuit.
    Now let us assume that it is true for $i$ and show that it is true for $i+1$. 
    The $(i+1)$-st $\psi$-gadget maps the state in \Cref{eq:ith-step} to
    \begin{align*}
        &\frac{1}{\sqrt{2^{i+1}}} \sum_{k=0}^i \ket{\psi_i(k)}_{\reg{CO}_{[1, i]}} \otimes \ket{0}_{\reg{C}_{[i+1]}}\ket{0}_{\reg{O}_{[i+1]}} \otimes \ket{0}_{\reg{CO}_{[i+2,n]}} \otimes (\ket{\psi}^{\otimes n -k} \otimes \ket{0}^{\otimes k})_{\reg{S}} \otimes \ket{k}_{\reg{K}}\\
        +{}& \frac{1}{\sqrt{2^{i+1}}} \sum_{k=0}^i \ket{\psi_i(k)}_{\reg{CO}_{[1, i]}} \otimes \ket{1}_{\reg{C}_{[i+1]}}\ket{\psi}_{\reg{O}_{[i+1]}} \otimes \ket{0}_{\reg{CO}_{[i+2,n]}} \otimes (\ket{\psi}^{\otimes n -k-1} \otimes \ket{0}^{\otimes k+1})_{\reg{S}} \otimes \ket{k+1}_{\reg{K}}.
    \end{align*}
    Rearranging, this is equal to
    \begin{equation*}
    \frac{1}{\sqrt{2^{i+1}}} \sum_{k=0}^{i+1}\ket{a_k} \otimes \ket{0}_{\reg{CO}_{[i+2,n]}} \otimes (\ket{\psi}^{\otimes n -k} \otimes \ket{0}^{\otimes k})_{\reg{S}} \otimes \ket{k}_{\reg{K}},
    \end{equation*}
    where
    \begin{equation*}
        \ket{a_k} = \left\{\begin{array}{ll}
                            \ket{\psi_i(0)}_{\reg{CO}_{[1, i]}} \otimes \ket{0}_{\reg{C}_{[i+1]}}\ket{0}_{\reg{O}_{[i+1]}} & \text{if $k = 0$,}\\
                            \ket{\psi_i(i)}_{\reg{CO}_{[1, i]}} \otimes \ket{1}_{\reg{C}_{[i+1]}}\ket{\psi}_{\reg{O}_{[i+1]}} & \text{if $k = i+1$,}\\
                            \ket{\psi_i(k)}_{\reg{CO}_{[1, i]}} \otimes \ket{0}_{\reg{C}_{[i+1]}}\ket{0}_{\reg{O}_{[i+1]}} + \ket{\psi_i(k-1)}_{\reg{CO}_{[1, i]}} \otimes \ket{1}_{\reg{C}_{[i+1]}}\ket{\psi}_{\reg{O}_{[i+1]}}& \text{otherwise.}
                            \end{array}\right.
    \end{equation*}
    It suffices to check that $\ket{a_k} = \ket{\psi_{i+1}(k)}$, and this is true by inspection. This completes the inductive step.
\end{proof}

This completes the proof of \Cref{lem:reference}, as
\begin{equation*}
    \tr_{\reg{SK}}(\ketbra{\circuit(\psi)}{\circuit(\psi)}) = \frac{1}{2^n} \sum_{k=0}^n \ketbra{\psi(k)}{\psi(k)},
\end{equation*}
due to the fact that the contents of the clock register are orthogonal for different values of $k$.

The circuit $\circuit$ consists of $n$ $\psi$-gadgets.
Each $\psi$-gadget involves one Hadamard gate, $n$ controlled $d$-dimensional $\mathrm{SWAP}$ gates, and one  $\lceil \log_2(n+1)\rceil$-qubit add-1 gate, for a total complexity of $\bigO{n \log(d)}$ gates per $\psi$-gadget.
This gives a total complexity of $\bigO{n^2 \log(d)}$ gates to compute the circuit $\circuit$.
\null\hfill\includegraphics[height=2.0ex]{acorn.png}

\subsection{Converting states to reflections about states}

\begin{lemma}[Simulating reflections from copies]
    \label{lem:reflect}
    Given $n = \bigO{1/\delta}$ copies of the state $\ket{\psi}$ and a state $\ket{\varphi}$, there is a circuit which outputs the reflection about $\ket{\psi}$ on $\ket{\varphi}$ to $\delta$ error.
    Specifically, this circuit performs a channel which is $\delta$-close to the reflection channel $\proj{\varphi} \mapsto (2\proj{\psi} - I) \proj{\varphi} (2 \proj{\psi} - I)$ in diamond distance.
    Moreover, this circuit is efficient and has gate complexity $\bigO{\log(d) /\delta}$.
\end{lemma}

This result is a corollary of quantum principal component analysis~\cite{lmr14}; we state the version of it due to Kimmel, Lin, Low, Ozols, and Yoder~\cite{klloy17}.
As we mention at the start of the section, similar statements have also appeared in prior work~\cite{aru14,jls18,gz25}, except with $n = \bigO{1/\delta^2}$.

\begin{lemma}[{Quantum principal component analysis \cite{lmr14}, as stated in \cite[Theorem 1]{klloy17}}]
    \label{lem:qpca}
    Let $\rho$ and $\sigma$ be two unknown quantum states, with $\sigma$ of dimension at least that of $\rho$.
    Let $t \in \R$.
    Then there exists a quantum algorithm that transforms $\sigma \otimes \rho^{\otimes n}$ into $\wt{\sigma}$ such that
    \begin{align*}
        \frac12 \trnorm{e^{-\ii \rho t} \sigma e^{\ii \rho t} - \wt{\sigma}} \leq \delta,
    \end{align*}
    as long as the number of copies of $\rho$ is $n = \bigO{t^2/\delta}$.
    In other words, this quantum algorithm implements the unitary $e^{-\ii \rho t}$ up to error $\delta$ in diamond norm, using $\bigO{t^2/\delta}$ copies of $\rho$.
\end{lemma}
\noindent \cref{lem:reflect} follows from applying \cref{lem:qpca} with $\rho \gets \proj{\psi}$, $\sigma \gets \proj{\varphi}$, and $t \gets \pi$, since
\begin{align*}
    e^{\ii \pi \proj{\varphi}} = e^{-\ii \pi \proj{\varphi}} = -(2\proj{\varphi} - I).
\end{align*}
By \cite[Remark 2]{klloy17}, this protocol uses $\bigO{\log(d) t^2/\delta}$ one- and two-qubit gates.

\subsection{Converting mixed state samples into random purifications}

\begin{lemma}[{Purifying a mixed state \includegraphics[height=2.0ex]{acorn.png}}]
    \label{lem:purify}
    There is a unitary circuit $\circuit^{(n)}$ such that, for all mixed states $\rho \in \C^{d \times d}$ of rank at most $r$,
    \begin{align}\label{eq:purification-formula}
        \circuit(\rho^{\otimes n}) = \E_{\ket{\varrho}} \proj{\varrho}^{\otimes n},
    \end{align}
    where $\ket{\varrho} \in \C^{d} \otimes \C^{r}$ is sampled uniformly from the space of purifications of $\rho$, meaning that $\tr_2(\proj{\varrho}) = \rho$.
    Moreover, this circuit is efficient and has a gate complexity of $\poly(n, \log(d), \log(1/\eps))$ to perform it to $\eps$ accuracy in diamond norm distance.
\end{lemma}

It is not possible to map copies of $\rho$ to copies of a single deterministic purification $\proj{\varrho}$~\cite{ldcl25}.
However, this result shows that we can convert copies of $\rho$ to copies of a random purification which is consistent across copies.
This purifies all of the ``relative'' randomness between copies of $\rho$, while leaving ``global'' randomness across all of the copies, which is far easier to handle computationally.
This result can be viewed as an algorithmic strengthening of the result \cite[Theorem 35]{sw22} and especially of the main result of~\cite{cwz24}. The work of \cite{cwz24} shows that when testing properties of the state $\rho$,
it does not help to be given copies of a random purification $\ket{\varrho}$.
This entails showing that a property testing algorithm which is given copies of a random purification $\ket{\varrho}$
can be simulated by a property testing algorithm which is only given copies of $\rho$,
and they do this by analyzing the acceptance probabilities of these two testers and showing that they coincide.
A key technical ingredient of this proof is explicit formulas for the mixed states $\rho^{\otimes n}$ and $\E_{\ket{\varrho}} \ketbra{\varrho}{\varrho}^{\otimes n}$.
Our contribution is in recognizing that these formulas reveal an efficient algorithm to convert between these two mixed states.
Therefore, we can use these ideas for algorithms and not just lower bounds.
We believe that our result gives a more conceptual reason for why random purifications do not help in property testing,
and we expect that our result will find applications in other domains as well.

\subsubsection{Representation theory background}

\newcommand{\epr}{\mathrm{EPR}}
\newcommand{\schur}{U_{\mathrm{Schur}}}
\newcommand{\specht}{\mathrm{Sp}}
\newcommand{\weyl}{\mathrm{V}}

In this subsection, we will cover standard representation theoretic material such as Schur-Weyl duality.
For a more thorough treatment of these topics, see, for example, \cite[Chapter 2]{Wri16}.

A \emph{partition of $n$} is a tuple of integers $\lambda = (\lambda_1, \ldots, \lambda_d)$ such that $\lambda_1 + \cdots + \lambda_d = n$ and $\lambda_1 \geq \cdots \geq \lambda_d \geq 0$.
We write this as $\lambda \vdash n$ for shorthand.
The \emph{length} $\ell(\lambda)$ of $\lambda$ is the number of nonzero entries in $\lambda$.
Given a partition $\lambda$, the corresponding \emph{Young diagram} is a collection of boxes arranged into $d$ rows in which the $i$-th row contains $\lambda_i$ boxes.
A \emph{standard Young tableau (SYT)} $S$ of shape $\lambda$ is a Young diagram of shape $\lambda$ whose entries have been filled in with the integers from $1$ through $n$, subject to the entries in each row being strictly increasing from left-to-right and the entries in each column being strictly increasing from top-to-bottom.
Similarly, a \emph{semistandard Young tableau (SSYT)} $T$ of shape $\lambda$ and alphabet $[d]$ is a Young diagram of shape $\lambda$ whose entries have been filled in with integers from the set $[d]$, subject to the rows being weakly increasing from left-to-right and the columns being strictly increasing from top-to-bottom.
We illustrate these definitions in \Cref{fig:young-tableau}.

\begin{figure*}[t!]
 \centering
    \begin{subfigure}[t]{0.3\textwidth}
        \centering
        \ytableausetup{centertableaux}
\begin{ytableau}
\phantom{1} & \phantom{1} & \phantom{1} & \phantom{1}  \\
\phantom{1} & \phantom{1} & \phantom{1} \\
\phantom{1}
\end{ytableau}
        \caption{The Young diagram $\lambda = (4, 3, 1)$. Note that $\lambda \vdash n$ for $n = 8$.}
    \end{subfigure}
    ~ 
    \begin{subfigure}[t]{0.3\textwidth}
        \centering
        \ytableausetup{centertableaux}
\begin{ytableau}
1 & 2 & 4 & 7  \\
3 & 6 & 8 \\
5
\end{ytableau}
\caption{A standard Young tableau of shape $\lambda$.}
    \end{subfigure}
    ~
    \begin{subfigure}[t]{0.3\textwidth}
        \centering
        \ytableausetup{centertableaux}
\begin{ytableau}
1 & 1 & 1 & 2  \\
2 & 2 & 3 \\
3
\end{ytableau}
\caption{A semistandard Young tableau of shape $\lambda$ and alphabet $[3]$.}
    \end{subfigure}
    \caption{}
    \label{fig:young-tableau}
\end{figure*}

\paragraph{Representation theory of the symmetric group.}
The irreducible representations of the symmetric group $S_n$ are indexed by partitions $\lambda \vdash n$ and are written $(\kappa_\lambda, \specht_{\lambda})$.
We will choose a particular basis of the irreducible representations of the symmetric group known as \emph{Young's orthogonal basis},
which gives an orthonormal basis for the space $\specht_{\lambda}$ consisting of vectors $\ket{S}$, where $S$ is an SYT of shape $\lambda$.
In this basis, the irreducible representations give rise to \emph{Young's orthogonal representation}, which are so-called because for each $\pi \in S_n$, $\kappa_{\lambda}(\pi)$ is an orthogonal matrix.
This implies the following property, which we will use in this work: when written in Young's orthogonal basis, each $\kappa_{\lambda}(\pi)$ matrix is real-valued.
For shorthand, we will write $\dim(\lambda) = \dim(\specht_{\lambda})$ for the dimension of the Specht module;
note that $\dim(\lambda)$ is equal to the size of Young's orthogonal basis, which is equal to the number of standard Young tableaus of shape $\lambda$.

\paragraph{Representation theory of the general linear group.}
The polynomial irreducible representations of the general linear group $GL(d)$ are indexed by partitions $\lambda$ in which $\ell(\lambda) \leq d$ and are written $(\nu^d_{\lambda}, \weyl_{\lambda}^d)$.
Here, ``polynomial'' refers to the fact that the matrix entries of $\nu^d_\lambda(M)$ are polynomials in the entries of $M$.
We will choose a particular basis of these irreducible representations of the general linear group known as the \emph{Gelfand-Tsetlin basis},
which gives an orthonormal basis for the space $\weyl_{\lambda}^d$ consisting of vectors $\ket{T}$, where $T$ is an SSYT of shape $\lambda$ and alphabet $[d]$.
A particularly important subgroup of the general linear group is the unitary group $U(d)$.
For each $\lambda$, $(\nu^d_{\lambda}, \weyl_{\lambda}^d)$ also serves as an irreducible representation of the unitary group.
For shorthand, we will write $s^d_{\lambda}(M) = \tr(\nu^d_{\lambda}(M))$, which is known as the \emph{Schur polynomial}.
We note that $s^d_{\lambda}(M)$ can be written as a polynomial in $M$'s eigenvalues, and it is zero unless $M$ has at least $\ell(\lambda)$ nonzero eigenvalues.

\paragraph{Schur-Weyl duality.}
There are two especially important representations of these groups which act on the space $(\C^d)^{\otimes n}$,
a representation $P(\pi)$ of the symmetric group
and a representation $Q^d(M)$ of the general linear group.
These representations are defined on standard basis vectors $\ket{i_1, \ldots, i_n} \in (\C^d)^{\otimes n}$ as follows:
\begin{align*}
    P(\pi) \cdot \ket{i_1, \ldots, i_n} &= \ket{i_{\pi^{-1}(1)}, \ldots, i_{\pi^{-1}(n)}}, \quad \text{for all $\pi \in S_n$},\\
    Q^d(M) \cdot \ket{i_1, \ldots, i_n} &= (M \cdot |i_1\rangle) \otimes \cdots \otimes (M \cdot |i_d\rangle), \quad \text{for all $U \in U(d)$}.
\end{align*}
These two representations commute,
and it turns out that they can be simultaneously diagonalized into irreducible representations in a particularly nice form.
This fact, known as \emph{Schur-Weyl duality},
states that there exists a unitary matrix $\schur$ known as the \emph{Schur transform} such that for any permutation $\pi \in S_n$ and unitary $M \in GL(d)$,
\begin{equation}\label{eq:schur-transform}
    \schur^d \cdot P(\pi)Q^d(M) \cdot (\schur^d)^{\dagger}
    = \sum_{\lambda \vdash n, \ell(\lambda) \leq d} \ketbra{\lambda}{\lambda} \otimes \kappa_{\lambda}(\pi) \otimes \nu^d_{\lambda}(M).
\end{equation}
We write $\schur^d$ to emphasize the dimensionality $d$, as our proof will involve two Schur transforms on registers of different dimensions,
though when the dimension $d$ is clear from context we will sometimes write it as $\schur$ for simplicity.
The Schur transform affects a change-of-basis from the standard basis $\ket{i_1, \ldots, i_n}$ into the \emph{Schur basis}, which consists of the vectors $\ket{\lambda, S, T}$, where $\lambda \vdash n$ is a partition with $\ell(\lambda) \leq d$, $S$ is an SYT of shape $\lambda$, and $T$ is an SSYT of shape $\lambda$ and alphabet $[d]$.
There is an efficient algorithm due to Bacon, Chuang, and Harrow~\cite{BCH05} which computes the Schur transform up to diamond distance $\epsilon$ in time $\poly(n, \log(d), \log(1/\epsilon))$. (The runtime of $\poly(n, d, \log(1/\epsilon))$ is shown in the original work, but a footnote at the bottom of \cite[Page 160]{Har05} sketches how to reduce the dependence on $d$ to $\poly\log(d)$; see \cite{bfglot25} for full details.)
We note that their work only shows an algorithm for the Schur transform that produces \emph{some} irrep matrices for the symmetric group.
In the work of Pelecanos, Spilecki, and Wright~\cite{psw25}, it is shown that their algorithm does indeed compute Young's orthogonal basis, and so the $\kappa_{\lambda}(\pi)$ matrices in \Cref{eq:schur-transform} do indeed correspond to Young's orthogonal form.

\subsubsection{Formulas for the mixed states}

The primary goal of this subsection is to compute a formula for the mixed state on the right-hand side of \Cref{eq:purification-formula} in the Schur basis.
To begin, let us note the following formula for the mixed state $\rho^{\otimes n}$ in the Schur basis:
\begin{equation*}
    \schur^d \cdot \rho^{\otimes n} \cdot (\schur^d)^{\dagger}
    = \schur^d \cdot Q^d(\rho) \cdot (\schur^d)^{\dagger}
    = \sum_{\lambda \vdash n, \ell(\lambda) \leq d} \ketbra{\lambda}{\lambda} \otimes I_{\dim(\lambda)} \otimes \nu^d_{\lambda}(\rho).
\end{equation*}
This equation follows from \Cref{eq:schur-transform} when all of $\rho$'s eigenvalues are positive, as in that case $\rho$ is an element of $GL(d)$.
When $\rho$ has eigenvalues which are 0, the right-hand side is at least still well-defined because the matrix entries of $\nu_{\lambda}(\rho)$ are polynomials in the entries of $\rho$.
And indeed the equation can still be seen to hold by writing $\rho$ as the limit of a set of states $\rho^+$ with positive eigenvalues which tend towards $\rho$ in the limit.
We note that because $\rho$ is PSD, $\nu_{\lambda}^d(\rho)$ must also be PSD.
Furthermore, if $\rho$ is rank $r$, then $s_{\lambda}^d(\rho) = \tr(\nu_{\lambda}^d(\rho)) = 0$ unless $\ell(\lambda) \leq r$, which implies that $\nu_{\lambda}^d(\rho) = 0$ unless $\ell(\lambda) \leq r$. As a result, we may restrict the sum over $\lambda$ to those of height $\ell(\lambda) \leq r$, as follows:
\begin{equation}\label{eq:rho-in-schur}
    \schur^d \cdot \rho^{\otimes n} \cdot (\schur^d)^{\dagger}
    = \sum_{\lambda \vdash n, \ell(\lambda) \leq r} \ketbra{\lambda}{\lambda} \otimes I_{\dim(\lambda)} \otimes \nu^d_{\lambda}(\rho).
\end{equation}

With this established, let us move towards a formula for the right-hand side of \Cref{eq:purification-formula} by first deriving a formula for a specific element of the mixture.
In other words, letting $\ket{\rho} \in \C^d \otimes \C^r$ be a fixed purification of $\rho$, we will derive a formula for $\ket{\rho}^{\otimes n}$ in the Schur basis.
To formalize this, note that each copy of $\ket{\rho}$ has two registers, the first of dimension $d$ which we will write as the $\reg{A}$ register and the second of dimension $r$ which we will write as the $\reg{A'}$ register. Then $\ket{\rho}^{\otimes n}$ has $n$ copies of the $\reg{A}$ register and $n$ copies of the $\reg{A'}$ register.
We will apply two separate Schur transforms, one to the $\reg{A}$ registers and the other to the $\reg{A'}$ registers, resulting in the state
\begin{equation*}
    (\schur ^d)_{\reg{A}} \otimes (\schur^r)_{\reg{A'}} \cdot \ket{\rho}^{\otimes n}.
\end{equation*}
Schur transforming the $\reg{A}$ registers will take us to the Schur basis, labeled by vectors of the form $\ket{\lambda}_{\reg{Y}} \ket{S}_{\reg{P}} \ket{T}_{\reg{Q}}$,
where $\reg{Y}$ is the ``Young diagram'' register, and $\reg{P}$ and $\reg{Q}$ are the symmetric and general linear group registers, respectively.
Similarly, Schur transforming the $\reg{A'}$ registers will take us to a Schur basis consisting of vectors of the form $\ket{\lambda'}_{\reg{Y'}} \ket{S'}_{\reg{P'}} \ket{T'}_{\reg{Q'}}$.
Note that $T$ is an SSYT of shape $\lambda$ and alphabet $[d]$, whereas $T'$ is an SSYT of shape $\lambda'$ and alphabet $[r]$.
Henceforth we will drop the $d$ and $r$ from the Schur transforms, with the understanding that $\schur^{\otimes 2}$ always refers to $(\schur ^d)_{\reg{A}} \otimes (\schur^r)_{\reg{A'}}$.

\paragraph{Step 1: permutation symmetry.}
The first property that we will note is that for any permutation $\pi \in S_n$,
\begin{equation*}
    P(\pi)_{\reg{A}} \otimes P(\pi)_{\reg{A'}} \cdot \ket{\rho}^{\otimes n} = \ket{\rho}^{\otimes n}.
\end{equation*}
Thus, by averaging over $\pi \in S_n$, we see that
\begin{equation}\label{eq:contained-in-proj}
        \E_{\pi \sim S_n} [P(\pi)_{\reg{A}} \otimes P(\pi)_{\reg{A'}}] \cdot \ket{\rho}^{\otimes n} = \ket{\rho}^{\otimes n}.
\end{equation}
Our first step to understanding $\ket{\rho}^{\otimes n}$,
then, involves understanding this average over permutations.
This is the focus of the next lemma.

\begin{definition}[Specht module EPR state]
    Let $\lambda \vdash n$. Then we write $\ket{\epr_{\lambda}}$ for the pure state inside $\specht_{\lambda} \otimes \specht_{\lambda}$ given by
    \begin{equation*}
        \ket{\epr_{\lambda}} \coloneqq \frac{1}{\sqrt{\dim(\lambda)}} \cdot \sum_{S} \ket{S} \otimes \ket{S},
    \end{equation*}
    where the sum ranges over all SYTs of shape $\lambda$.
\end{definition}

\begin{lemma}[Averaging over permutations gives an EPR state]\label{lem:averaging-epr}
    \begin{equation*}
        \schur^{\otimes 2} \cdot \Big(\E_{\pi \sim S_n} P(\pi)_{\reg{A}} \otimes P(\pi)_{\reg{A'}}\Big) \cdot (\schur^{\dagger})^{\otimes 2}
        = \sum_{\lambda \vdash n, \ell(\lambda) \leq r} \ketbra{\lambda\lambda}{\lambda\lambda}_{\reg{Y}\reg{Y'}} \otimes \ketbra{\epr_{\lambda}}{\epr_{\lambda}}_{\reg{P}\reg{P'}} \otimes I_{\reg{Q}\reg{Q'}}.
    \end{equation*}
\end{lemma}
\begin{proof}
    By Schur-Weyl duality,
    \begin{align}
        &\schur^{\otimes 2} \cdot \Big(\E_{\pi \sim S_n} P(\pi)^{\otimes 2}\Big) \cdot (\schur^{\dagger})^{\otimes 2}\nonumber\\
        ={}& \E_{\pi \sim S_n}\Big[\sum_{\lambda\vdash n,\ell(\lambda)\leq d}\sum_{\mu \vdash n,\ell(\mu)\leq r} \ketbra{\lambda}{\lambda}_{\reg{Y}} \otimes \ketbra{\mu}{\mu}_{\reg{Y'}} \otimes \kappa_{\lambda}(\pi)_{\reg{P}} \otimes \kappa_{\mu}(\pi)_{\reg{P'}} \otimes (I_{\dim(V_{\lambda}^d)})_{\reg{Q}} \otimes (I_{\dim(V_{\mu}^r)})_{\reg{Q'}}\Big]\nonumber\\
        ={}& \sum_{\lambda\vdash n,\ell(\lambda)\leq d}\sum_{\mu \vdash n,\ell(\mu)\leq r} \ketbra{\lambda}{\lambda}_{\reg{Y}} \otimes \ketbra{\mu}{\mu}_{\reg{Y'}} \otimes \E_{\pi \sim S_n}\big[\kappa_{\lambda}(\pi)_{\reg{P}} \otimes \kappa_{\mu}(\pi)_{\reg{P'}}\big] \otimes (I_{\dim(V_{\lambda}^d)})_{\reg{Q}} \otimes (I_{\dim(V_{\mu}^r)})_{\reg{Q'}}.\label{eq:for-later}
    \end{align}
    Expanding the expectation in Young's orthogonal basis, we have
    \begin{align}
        \E_{\pi \sim S_n}\big[\kappa_{\lambda}(\pi)_{\reg{P}} \otimes \kappa_{\mu}(\pi)_{\reg{P'}}\big]
        &= \sum_{S_1,S_2, S_3, S_4} \ketbra{S_1}{S_2} \otimes \ketbra{S_3}{S_4} \cdot \E_{\pi \sim S_n}[\kappa_{\lambda}(\pi)_{S_1, S_2} \cdot \kappa_{\mu}(\pi)_{S_3, S_4}]\nonumber\\
        &= \sum_{S_1,S_2, S_3, S_4} \ketbra{S_1}{S_2} \otimes \ketbra{S_3}{S_4} \cdot \E_{\pi \sim S_n}[\overline{\kappa_{\lambda}(\pi)_{S_1, S_2}} \cdot \kappa_{\mu}(\pi)_{S_3, S_4}],\label{eq:just-conjugated}
    \end{align}
    where in the last step we used the fact that $\kappa_{\lambda}(\pi)$ is a real-valued matrix, and so its entries are equal to their own conjugates. In these expressions, $S_1$ and $S_2$ are SYTs of shape $\lambda$, and $S_3$ and $S_4$ are SYTs of shape~$\mu$.
    Then the grand Schur orthogonality relations tell us that
    \begin{equation*}
        \E_{\pi \sim S_n}[\overline{\kappa_{\lambda}(\pi)_{S_1, S_2}} \cdot \kappa_{\mu}(\pi)_{S_3, S_4}]
        = \left\{\begin{array}{cl}
                1/\dim(\lambda) & \text{if $\lambda = \mu$, $S_1 = S_3$, $S_2 = S_4$},\\
                0 & \text{otherwise}.
                \end{array}\right.
    \end{equation*}
    Plugging this in above, we have that \Cref{eq:just-conjugated} is 0 unless $\lambda = \mu$, in which case it is
    \begin{equation*}
        \E_{\pi \sim S_n}\big[\kappa_{\lambda}(\pi)_{\reg{P}} \otimes \kappa_{\mu}(\pi)_{\reg{P'}}\big]
        = \sum_{S_1, S_2} \ketbra{S_1}{S_2} \otimes \ketbra{S_1}{S_2} \cdot \frac{1}{\dim(\lambda)}
        = \ketbra{\epr_{\lambda}}{\epr_{\lambda}}.
    \end{equation*}
    As a result,
    \begin{equation*}
        \eqref{eq:for-later}
        = \sum_{\lambda\vdash n, \ell(\lambda) \leq r} \ketbra{\lambda}{\lambda}_{\reg{Y}} \otimes \ketbra{\lambda}{\lambda}_{\reg{Y'}} \otimes \ketbra{\epr_{\lambda}}{\epr_{\lambda}}_{\reg{P}\reg{P'}} \otimes (I_{\dim(V_{\lambda}^d)})_{\reg{Q}} \otimes (I_{\dim(V_{\lambda}^r)})_{\reg{Q'}}
    \end{equation*}
    This completes the proof.
\end{proof}

Now we can use this to characterize $\ket{\rho}^{\otimes n}$, as follows.
\begin{lemma}\label{lem:purification-formula}
    For any purification $\ket{\rho}$ of $\rho$, we have
    \begin{equation*}
    (\schur)_{\reg{A}} \otimes (\schur)_{\reg{A'}} \cdot \ket{\rho}^{\otimes n}
    = \sum_{\lambda \vdash n, \ell(\lambda)\leq r}\sum_{T, T'} a_{\lambda, T, T'} \cdot \ket{\lambda, \lambda}_{\reg{Y} \reg{Y'}} \ket{\epr}_{\reg{P}\reg{P'}} \ket{T}_{\reg{Q}} \ket{T'}_{\reg{Q'}},
\end{equation*}
for some complex coefficients $a_{\lambda, T, T'}$.
\end{lemma}
\begin{proof}
    If we write $\Pi = \E_{\pi \sim S_n} P(\pi)_{\reg{A}} \otimes P(\pi)_{\reg{A'}}$, \Cref{lem:averaging-epr} tells us that $\Pi$ is a projection matrix onto the subspace of states spanned, in the Schur basis, by vectors of the form 
    \begin{equation*}
        \ket{\lambda\lambda}_{\reg{Y}\reg{Y'}} \ket{\epr_{\lambda}}_{\reg{P}\reg{P'}}  \ket{T}_{\reg{Q}}  \ket{T'}_{\reg{Q'}},
    \end{equation*}
    where $\ell(\lambda) \leq r$.
    Furthermore, \Cref{eq:contained-in-proj} tells us that $\ket{\rho}^{\otimes n}$ is entirely contained inside $\Pi$, and so it can be written as a linear combination of states of this form. This proves the lemma.
\end{proof}

\paragraph{Step 2: unitary symmetry.}
The second property we will note is that the state $\E_{\ket{\rho}} \ketbra{\rho}{\rho}^{\otimes n}$ has unitary symmetry. In particular, writing $U$ for a Haar random unitary in $U(r)$, we note that
\begin{equation}\label{eq:randomized-fixed-guy}
    \E_{\ket{\rho}} \ketbra{\rho}{\rho}^{\otimes n}
    = \E_{U \sim U(r)} Q^r(U)_{\reg{A'}} \cdot \ketbra{\rho_0}{\rho_0} \cdot Q^r(U)_{\reg{A'}}^{\dagger},
\end{equation}
for any fixed purification $\ket{\rho_0}$ of $\rho$.
This follows because if $\rho$ has eigendecomposition $\rho = \sum_{i=1}^r \alpha_i\cdot  \ketbra{u_i}{u_i}$, then the Schmidt decomposition of any purification of $\rho$ can be written as
\begin{equation*}
    \sum_{i=1}^r \sqrt{\alpha_i} \cdot \ket{u_i} \otimes \ket{v_i},
\end{equation*}
where the right Schmidt vectors $\ket{v_1}, \ldots, \ket{v_r}$ form an orthonormal basis of $\C^r$.
A random purification involves choosing the right Schmidt vectors to be a Haar random basis of $\C^r$, and this is equivalent to picking a fixed basis of $\C^r$ (say, the right Schmidt vectors of $\ket{\rho_0}$) and applying a Haar random rotation $U$.
Using this, we derive a formula for $\E_{\ket{\rho}} \ketbra{\rho}{\rho}^{\otimes n}$.

\begin{lemma}[Haar averaging]\label{lem:haar-average}
There exist matrices $\{M_{\lambda}\}_{\lambda}$ such that
    \begin{equation*}
        \schur^{\otimes 2} \cdot\E_{\ket{\rho}} \ketbra{\rho}{\rho}^{\otimes n} \cdot (\schur^\dagger)^{\otimes 2}
        = \sum_{\lambda \vdash n, \ell(\lambda) \leq r} \ketbra{\lambda, \lambda}{\lambda, \lambda}_{\reg{Y} \reg{Y'}} \otimes \ketbra{\epr_{\lambda}}{\epr_{\lambda}}_{\reg{P} \reg{P'}} \otimes (M_{\lambda})_{\reg{Q}} \otimes (I_{\dim(V_{\lambda}^r)})_{\reg{Q'}}.
    \end{equation*}
\end{lemma}
\begin{proof}
    From \Cref{lem:purification-formula},
    we see that in the Schur basis, $\ketbra{\rho_0}{\rho_0}^{\otimes n}$ can be written as a linear combination of terms of the form
    \begin{equation*}
        \ketbra{\lambda, \lambda}{\mu, \mu}_{\reg{Y} \reg{Y'}} \otimes \ketbra{\epr_{\lambda}}{\epr_{\mu}}_{\reg{P} \reg{P'}} \otimes \ketbra{T}{T''}_{\reg{Q}}\otimes \ketbra{T'}{T'''}_{\reg{Q'}},
    \end{equation*}
    where $\ell(\lambda)$ and $\ell(\mu) \leq r$,
    $T$ and $T'$ are SSYTs of shape $\lambda$ (and alphabets $[d]$ and $[r]$, respectively), and $T''$ and $T'''$ are SSYTs of shape $\mu$ (and alphabets $[d]$ and $[r]$, respectively).
    Now applying $Q^r(U)_{\reg{A'}}$ for a random $U \sim U(r)$, this becomes
    \begin{equation*}
        \ketbra{\lambda, \lambda}{\mu, \mu}_{\reg{Y} \reg{Y'}} \otimes \ketbra{\epr_{\lambda}}{\epr_{\mu}}_{\reg{P} \reg{P'}} \otimes \ketbra{T}{T''}_{\reg{Q}}\otimes \E_{U \sim U(r)}[\nu^r_{\lambda}(U) \cdot \ketbra{T'}{T'''}_{\reg{Q'}} \cdot \nu^r_{\mu}(U)^{\dagger}].
    \end{equation*}
    But Schur's lemma tells us that
    \begin{equation*}
        \E_{U \sim U(d)}[\nu^r_{\lambda}(U) \cdot \ketbra{T'}{T'''}_{\reg{Q'}} \cdot \nu^r_{\mu}(U)^{\dagger}]
    \end{equation*}
    is always a multiple of $I_{\dim(V_{\lambda}^r)}$ if $\lambda = \mu$, and is 0 otherwise.
    Thus, \Cref{eq:randomized-fixed-guy}, in the Schur basis, can be written as a linear combination of terms of the form
    \begin{equation*}
        \ketbra{\lambda, \lambda}{\lambda, \lambda}_{\reg{Y} \reg{Y'}} \otimes \ketbra{\epr_{\lambda}}{\epr_{\lambda}}_{\reg{P} \reg{P'}} \otimes \ketbra{T}{T'}_{\reg{Q}} \otimes (I_{\dim(V_{\lambda}^r)})_{\reg{Q'}},
    \end{equation*}
    where $T$ and $T'$ are both SYTs of shape $\lambda$ and dimension $[d]$.
    In other words, there exist coefficients $a_{\lambda, T, T'}$ such that
    \begin{align*}
        &\schur^{\otimes 2} \cdot \Big(\E_{\ket{\rho}} \ketbra{\rho}{\rho}^{\otimes n}\Big) \cdot (\schur^\dagger)^{\otimes 2}\\
         ={}& \sum_{\lambda, T, T'} a_{\lambda, T, T'} \cdot \ketbra{\lambda, \lambda}{\lambda, \lambda}_{\reg{Y} \reg{Y'}} \otimes \ketbra{\epr_{\lambda}}{\epr_{\lambda}}_{\reg{P} \reg{P'}} \otimes \ketbra{T}{T'}_{\reg{Q}} \otimes (I_{\dim(V_{\lambda}^r)})_{\reg{Q'}}\\
         ={}& \sum_{\lambda} \ketbra{\lambda, \lambda}{\lambda, \lambda}_{\reg{Y} \reg{Y'}} \otimes \ketbra{\epr_{\lambda}}{\epr_{\lambda}}_{\reg{P} \reg{P'}} \otimes \Big(\sum_{T, T'} a_{\lambda, T, T'} \cdot\ketbra{T}{T'}_{\reg{Q}}\Big) \otimes (I_{\dim(V_{\lambda}^r)})_{\reg{Q'}}.
    \end{align*}
    Defining $M_{\lambda} \coloneqq \sum_{T, T'} a_{\lambda, T, T'} \cdot \ketbra{T}{T'}$ completes the proof.
\end{proof}

\paragraph{Step 3: finishing up.}
Now that we have derived a generic formula in \Cref{lem:haar-average} for states which exhibit permutation and unitary symmetries, we now specialize to our case by using the fact that our states are purifications of $\rho$.
This gives our desired formula for the mixture.
\begin{lemma}[Mixture formula]\label{lem:final-formula}
\begin{equation*}
    \schur^{\otimes 2} \cdot\E_{\ket{\rho}} \ketbra{\rho}{\rho}^{\otimes n} \cdot (\schur^\dagger)^{\otimes 2}
        = \sum_{\ell(\lambda)\leq r} \dim(\lambda) \cdot \ketbra{\lambda, \lambda}{\lambda, \lambda}_{\reg{Y} \reg{Y'}} \otimes \ketbra{\epr_{\lambda}}{\epr_{\lambda}}_{\reg{P}\reg{P'}} \otimes \nu^d_{\lambda}(\rho)_{\reg{Q}} \otimes \Big(\frac{I_{\dim(V_{\lambda}^r)}}{\dim(V_{\lambda}^r)}\Big)_{\reg{Q'}}.
\end{equation*}
\end{lemma}
\begin{proof}
    Since $\ket{\rho}$ is a purification of $\rho$,
    we have that $\tr_{\reg{A'}}(\ketbra{\rho}{\rho}^{\otimes n}) = \rho^{\otimes n}$.
    Applying the partial trace to the mixture, we have
    \begin{equation*}
        \tr_{\reg{A'}}\Big(\E_{\ket{\rho}}\ketbra{\rho}{\rho}^{\otimes n}\Big) = \rho^{\otimes n}.
    \end{equation*}
    In the Schur basis, we have seen in \Cref{eq:rho-in-schur} that
    \begin{equation}\label{eq:repeat-equations}
    \schur \cdot \rho^{\otimes n} \cdot \schur^{\dagger}
    = \sum_{\lambda \vdash n, \ell(\lambda) \leq r} \ketbra{\lambda}{\lambda} \otimes I_{\dim(\lambda)} \otimes \nu^d_{\lambda}(\rho).
    \end{equation}
    On the other hand, \Cref{lem:haar-average} gives a formula for the mixture $\ketbra{\rho}{\rho}^{\otimes n}$ in the Schur basis. Tracing out the $\reg{A'}$ registers corresponds, in the Schur basis, to tracing out the $\reg{Y'}$, $\reg{P'}$, and $\reg{Q'}$ registers, yielding the state
    \begin{equation}\label{eq:traced-out}
        \tr_{\reg{A'}}\Big(\schur^{\otimes 2} \cdot\E_{\ket{\rho}} \ketbra{\rho}{\rho}^{\otimes n} \cdot (\schur^\dagger)^{\otimes 2}\Big)
        = \sum_{\ell(\lambda)\leq r} \ketbra{\lambda}{\lambda}_{\reg{Y}} \otimes \Big(\frac{I_{\dim(\lambda)}}{\dim(\lambda)}\Big)_{\reg{P}} \otimes (M_{\lambda})_{\reg{Q}} \cdot \dim(V_{\lambda}^r),
    \end{equation}
    where here we have used the fact that $\tr_{\reg{P'}}(\ketbra{\epr_{\lambda}}{\epr_{\lambda}}_{\reg{P}\reg{P'}}) = I_{\dim(\lambda)}/\dim(\lambda)$.
    Now, \Cref{eq:repeat-equations,eq:traced-out}
    must be equal to each other, which implies that
    \begin{equation*}
        M_{\lambda} = \nu^d_{\lambda}(\rho) \cdot \frac{\dim(\lambda)}{\dim(V_{\lambda}^r)}.
    \end{equation*}
    Plugging this back into \Cref{lem:haar-average} completes the proof.
\end{proof}

\subsubsection{The random purification algorithm}

Now we give our algorithm for producing random purifications.
Given an input state $\rho^{\otimes n}$, our algorithm works as follows.
\begin{enumerate}
    \item Apply the Schur transform $\schur$ to $\rho^{\otimes n}$.
    \item Perform the projective measurement $\{\Pi_{\lambda}\}_{\lambda}$, where $\Pi_{\lambda} = \ketbra{\lambda}{\lambda}_{\reg{Y}} \otimes I_{\reg{P}} \otimes I_{\reg{Q}}$.
    Let $\lambda$ with $\ell(\lambda) \leq r$ be the outcome.
    \item Introduce a new register $\reg{Y'}$ and copy the contents of $\reg{Y}$ into it.
    \item Introduce a new register $\reg{P'}$. Discard the contents of the $\reg{P}$ register and reinitialize the two registers with an $\ket{\epr_{\lambda}}_{\reg{P} \reg{P'}}$ state.
    \item Introduce a new register $\reg{Q'}$ initialized to the maximally mixed state $I_{\dim(V_{\lambda}^r)}/\dim(V_{\lambda}^r)$.
    \item Apply the inverse Schur transform $(\schur^{d,\dagger})_{\reg{A}} \otimes (\schur^{r,\dagger})_{\reg{A'}}$ and output the result.
\end{enumerate}
Now we use this to prove our \Cref{lem:purify}.

\renewcommand{\qedsymbol}{\includegraphics[height=2.0ex]{acorn.png}}
\begin{proof}[Proof of \Cref{lem:purify}]
To check that our algorithm produces the correct output state, let us track how the input state changes at each step. 
From \Cref{eq:rho-in-schur}, after the Schur transform in step 1, it changes to
\begin{equation*}
    \sum_{\lambda \vdash n, \ell(\lambda) \leq r} \ketbra{\lambda}{\lambda}_{\reg{Y}} \otimes (I_{\dim(\lambda)})_{\reg{P}} \otimes \nu^d_{\lambda}(\rho)_{\reg{Q}}.
\end{equation*}
Then the measurement in step 2 produces the partition $\lambda$ with probability $\dim(\lambda) \cdot s_{\lambda}(\rho)$,
and the state collapses to
\begin{equation*}
    \ketbra{\lambda}{\lambda}_{\reg{Y}} \otimes \Big(\frac{I_{\dim(\lambda)}}{\dim(\lambda)}\Big)_{\reg{P}} \otimes \Big(\frac{\nu^d_{\lambda}(\rho)}{s^d_{\lambda}(\rho)}\Big)_{\reg{Q}}.
\end{equation*}
From here, the state transforms as
\begin{align*}
    \stackrel{\text{step 3}}{\longrightarrow}{}&\ketbra{\lambda, \lambda}{\lambda, \lambda}_{\reg{Y} \reg{Y'}} \otimes \Big(\frac{I_{\dim(\lambda)}}{\dim(\lambda)}\Big)_{\reg{P}} \otimes \Big(\frac{\nu^d_{\lambda}(\rho)}{s^d_{\lambda}(\rho)}\Big)_{\reg{Q}}\\
    \stackrel{\text{step 4}}{\longrightarrow}{}&\ketbra{\lambda, \lambda}{\lambda, \lambda}_{\reg{Y} \reg{Y'}} \otimes \ketbra{\epr}{\epr}_{\reg{P} \reg{P'}} \otimes \Big(\frac{\nu^d_{\lambda}(\rho)}{s^d_{\lambda}(\rho)}\Big)_{\reg{Q}}\\
    \stackrel{\text{step 5}}{\longrightarrow}{}&\ketbra{\lambda, \lambda}{\lambda, \lambda}_{\reg{Y} \reg{Y'}} \otimes \ketbra{\epr}{\epr}_{\reg{P} \reg{P'}} \otimes \Big(\frac{\nu^d_{\lambda}(\rho)}{s^d_{\lambda}(\rho)}\Big)_{\reg{Q}} \otimes \Big(\frac{I_{\dim(V_{\lambda}^r)}}{\dim(V_{\lambda}^r)}\Big)_{\reg{Q'}}.
\end{align*}
After the Schur transform in step 6, the final state is
\begin{equation*}
    (\schur^{\dagger})^{\otimes 2}\cdot \Big(\ketbra{\lambda, \lambda}{\lambda, \lambda}_{\reg{Y} \reg{Y'}} \otimes \ketbra{\epr}{\epr}_{\reg{P} \reg{P'}} \otimes \Big(\frac{\nu^d_{\lambda}(\rho)}{s^d_{\lambda}(\rho)}\Big)_{\reg{Q}}  \otimes \Big(\frac{I_{\dim(V_{\lambda}^r)}}{\dim(V_{\lambda}^r)}\Big)_{\reg{Q'}} \Big) \cdot  \schur^{\otimes 2},
\end{equation*}
with probability $\dim(\lambda) \cdot s_{\lambda}(\rho)$.
The overall mixed state, then, that the algorithm outputs is
\begin{equation*}
    (\schur^{\dagger})^{\otimes 2}\cdot \Big(\sum_{\lambda\vdash n, \ell(\lambda) \leq r} \dim(\lambda) \cdot \ketbra{\lambda, \lambda}{\lambda, \lambda}_{\reg{Y} \reg{Y'}} \otimes \ketbra{\epr}{\epr}_{\reg{P} \reg{P'}} \otimes \nu^d_{\lambda}(\rho)_{\reg{Q}} \otimes \Big(\frac{I_{\dim(V_{\lambda}^r)}}{\dim(V_{\lambda}^r)}\Big)_{\reg{Q'}} \Big) \cdot  \schur^{\otimes 2}.
\end{equation*}
By \Cref{lem:final-formula}, this is equal to $\E_{\ket{\rho}} \ketbra{\rho}{\rho}^{\otimes n}$, as desired.
The gate complexity of this algorithm is dominated by the two Schur transforms, which take $\poly(n, \log(d), \log(1/\epsilon))$ gates to compute to $\epsilon$ accuracy.
This completes the proof.
\end{proof}
\renewcommand{\qedsymbol}{$\square$}

\subsection{Putting everything together}

\begin{proof}[Proof of \cref{thm:useless}]
Our simulation algorithm proceeds in the following way.
Recall that we are given $n = \bigO{q^2/\eps}$ copies of $\sigma$, which is rank at most $r$.

\begin{enumerate}
    \item Apply \cref{lem:purify} to convert our $n$ copies of $\sigma$ to $n$ copies of $\proj{\varrho}$, where $\varrho$ is a random $rd$-dimensional purification of $\sigma$;
    \item Apply \cref{lem:reference} to convert our $n$ copies of $\proj{\varrho}$ to $n$ copies of $\proj{\varsigma}$, where $\ket{\varsigma} = \frac{1}{\sqrt{2}}(e^{\ii \theta}\ket{0}\ket{0} + \ket{1}\ket{\varrho})$ for a uniformly random $e^{\ii\theta}$;
    \item Run the circuit, except for every application of $U$, instead the technique from \cref{lem:reflect} to simulate a reflection about $\ket{\varsigma}$ to error $\eps / q$.
\end{enumerate}
Every approximate reflection uses $\bigO{q/\eps}$ copies of $\proj{\varsigma}$, giving the desired final sample complexity of $\bigO{q^2/\eps}$.

The only error in this algorithm is in the use of approximate reflections: each one incurs $\eps/q$ error in diamond distance, so by the triangle inequality, the final output will be off by at most $\eps$ in trace distance to the circuit run on $V_{\ket{\varsigma}} \coloneqq 2\proj{\varsigma} - I$, the reflection about $\ket{\varsigma}$.

We are done upon proving that $V_{\ket{\varsigma}}$ is a state preparation unitary for $\sigma$.
Notice that
\begin{align*}
    V_{\ket{\varsigma}}\ket{0}\ket{0}
    = 2 \frac{e^{-\ii \theta}}{\sqrt{2}} \ket{\varsigma} - \ket{0}\ket{0}
    = e^{-\ii \theta} \ket{1}\ket{\varrho}.
\end{align*}
Tracing out the first qubit and the second register of $\ket{\varrho}$ gives us the density matrix $\rho$.
So, following the definition \cref{def:state-prep}, $V_{\ket{\varsigma}}$ is a state preparation unitary of $\sigma$.
This state preparation unitary is of size $2rd \times 2rd$; by padding, we can make this any size $\wh{d}d \times \wh{d} d$ for $\wh{d} \geq 2r$.
\end{proof}

\section{Conjugate queries help for reality testing}

In this section, we prove \cref{thm:useful}.

First, we observe that the function we wish to test, $\real(\ket{\psi})$, is indeed related to the closeness of $\ket{\psi}$ to a real vector, up to a global phase.
Reality testing, then, neatly corresponds to a property testing problem in the traditional sense~\cite{MdW16}: detect whether a state is in the space of states with real amplitudes, or far from it, in the sense of fidelity.\footnote{
    In principle, we could extend this quantity to general mixed states, by taking $\real(\rho)$ to be the fidelity between $\rho$ and $\rho^*$, or perhaps just $\tr(\rho^* \rho)$.
    When we do this, we lose the interpretation of measuring ``closeness to reality'', since the quantity can be small for very mixed classical states.
}

\begin{fact}[Reality testing tests closeness to the space of real states]
    \label{fact:closeness}
    Let $\ket{\psi} \in \C^d$ be a unit vector.
    Then
    \begin{align*}
        \max_{\ket{\phi} \in \R^d} \abs{\braket{\phi}{\psi}}^2 = \max_{\theta \in [0, 2\pi)} \norm[\Big]{\Re(e^{\ii \theta}\ket{\psi})}^2
        = \frac12 + \frac12 \sqrt{\real(\ket{\psi})}
    \end{align*}
    where $\Re(\ket{\psi})$ denotes the real part of the vector, or equivalently, the projection of $\ket{\psi}$ onto the subspace of real vectors, treating $\C^d$ as a real vector space of dimension $2d$.
\end{fact}
\begin{proof}
We write $\ket{\psi} = \sum_{j=1}^d (a_j + \ii b_j) \ket{j}$, where $\ket{a}, \ket{b} \in \R^d$ are vectors with real coefficients.
Note that $\braket{a}{a} + \braket{b}{b} = \braket{\psi}{\psi} = 1$.
Then, for a vector $\ket{\phi}$ with real amplitudes,
\begin{align*}
    \abs{\braket{\phi}{\psi}}^2
    &= \abs{\braket{\phi}{a} + \ii\braket{\phi}{b}}^2 \\
    &= \abs{\braket{\phi}{a}}^2 + \abs{\braket{\phi}{b}}^2 \\
    &= \bra{\phi}\parens[\big]{\proj{a} + \proj{b}}\ket{\phi}.
\end{align*}
So, the maximum value achieved by $\ket{\phi}$ is equal to the top eigenvalue of $\proj{a} + \proj{b}$, which can be computed to be (provided that $\braket{a}{a} + \braket{b}{b} = 1$)
\begin{equation}
    \label{eq:closeness}
    \frac{1}{2} + \sqrt{\frac{(\braket{a}{a} - \braket{b}{b})^2}{4} + \braket{a}{b}^2}.
\end{equation}
Now, for the second expression in the statement,
\begin{align*}
    \Re(e^{\ii \theta} \ket{\psi})
    &= \Re((\cos(\theta) + \ii \sin(\theta))\ket{\psi})
    = \sum_{j=1}^d (a_j \cos(\theta) - b_j \sin(\theta)) \ket{j} \\
    \norm[\big]{\Re(e^{\ii \theta} \ket{\psi})}^2
    &= \sum_{j=1}^d (a_j \cos(\theta) - b_j \sin(\theta))^2 \\
    &= \braket{a}{a} \cos^2(\theta) + \braket{b}{b} \sin^2(\theta) - \braket{a}{b} 2 \sin(\theta) \cos(\theta) \\
    &= \frac{\braket{a}{a} + \braket{b}{b}}{2} + \frac{\braket{a}{a} - \braket{b}{b}}{2} \cos(2\theta) - \braket{a}{b} \sin(2 \theta)
\end{align*}
The final step uses the double-angle formulas $\cos(2 \theta) = \cos^2(\theta) - \sin^2(\theta)$ and $\sin(2 \theta) = 2 \sin(\theta) \cos(\theta)$.
The maximum value this attains by varying $\theta$ is identical to \eqref{eq:closeness}.
We can conclude by observing that
\begin{align*}
    \real(\ket{\psi}) = \abs[\Big]{\sum_{j=1}^d (a_j + \ii b_j)^2}^2
    = \abs{\braket{a}{a} - \braket{b}{b} + 2 \ii \braket{a}{b}}^2
    = (\braket{a}{a} - \braket{b}{b})^2 + (2 \braket{a}{b})^2,
\end{align*}
which is indeed the expression under the radical in \eqref{eq:closeness}, up to a factor of 4.
\end{proof}

\subsection{Algorithm using conjugate queries}

\begin{lemma}[Reality testing algorithm]
    \label{lem:upper}
    Let $U$ be a state preparation unitary for $\proj{\psi}$.
    Then reality testing (\cref{prob:real-test}) can be solved using $2$ queries to $U$ and $2$ queries to $U^*$.

    In general, given an error parameter $\eps > 0$, with $\bigO{1/\eps^2}$ queries to $U$ and $U^*$, $\real(\ket{\psi})$ can be estimated to $\eps$ additive error with success probability $\geq 2/3$.
\end{lemma}
\begin{proof}
The algorithm is a swap test: apply $U$ on one register to prepare a copy of $\proj{\psi}$, and apply $U^*$ on another register to prepare a copy of $\proj{\psi^*}$.
Then, apply the swap test on the two registers:
\begin{figure}[h]
\begin{equation*}
    \Qcircuit @C=1em @R=.3em @!C @!R {
    \lstick{\ket{0}} & \gate{H} & \ctrl{2} & \gate{H} & \meter \\
    \lstick{\ket{\psi}} & \qw & \qswap & \qw & \qw \\
    \lstick{\ket{\psi^*}} & \qw & \qswap & \qw & \qw
}
\end{equation*}
\end{figure}
The probability that the measurement outcome is $0$ is $\frac12 + \frac12 \abs{\braket{\psi}{\psi^*}}^2 = \frac12 + \frac12 \real(\ket{\psi})$.
Run this circuit twice; if the measurement outcome is $0$ both times, guess that $\real(\ket{\psi}) = 1$; otherwise, guess that $\real(\ket{\psi}) < 1/10$.

If $\real(\ket{\psi}) = 1$, then the measurement outcome will always be $0$, so this algorithm always succeeds.
If $\real(\ket{\psi}) < 1/10$, then the probability that both outcomes are 0 is at most $(0.55)^2 < 1/3$.
So, in this case, the algorithm succeeds with probability $\geq 2/3$.

For the estimation task, if we run the swap test $\bigO{1/\eps^2}$ times, then by standard methods, we can produce an estimate this probability $\frac12 + \frac12 \real(\ket{\psi})$ to $\eps/2$ additive error with probability $\geq 2/3$.
This estimator is just the fraction of measurement outcomes which are $0$.
This gives the desired estimate for $\real(\ket{\psi})$ to $\eps$ error.
\end{proof}

\subsection{Lower bound without conjugate queries}

\begin{problem}[Distinguishing phase states from Haar-random states] \label{prob:phase-prs}
    Let $\calF$ denote the ensemble of phase states in $d$ dimensions:
    \begin{align*}
        \calF = \braces[\Big]{\ket{f} = \frac{1}{\sqrt{d}}\sum_{i = 1}^d (-1)^{f(i)} \ket{i} \Big| f: [d] \to \braces{0, 1}}.
    \end{align*}
    Choose a bit $b \in \braces{0,1}$.
    If $b = 0$, draw a state $\ket{\psi}$ from the Haar-random distribution; if $b = 1$, draw $\ket{\psi}$ uniformly at random from $\calF$.
    The goal is to, given access to $\ket{\psi}$ (either copies of the state or access to a state preparation unitary), output a bit $\wh{b} \in \braces{0,1}$ such that $\wh{b} = b$ with probability $\geq 0.65$.
\end{problem}

\begin{lemma}
    Suppose $d > 10^3$.
    Then \cref{prob:phase-prs} can be solved with one call to an algorithm for reality testing (\cref{prob:real-test}).
\end{lemma}
\begin{proof}
Let $\ket{\psi}$ be the input state; we want to decide whether it is a phase state or Haar-random.
For all $\ket{\psi} \in \calF$, $\real(\ket{\psi}) = 1$, since its amplitudes are real.

On the other hand, for a Haar-random state $\ket{\psi}$,
\begin{align*}
    \E\bracks[\big]{\real(\ket{\psi})}
    = \E\bracks[\Big]{\abs[\Big]{\sum_{i=1}^d \psi_i^2}^2}
    = \sum_{i=1}^d \sum_{j=1}^d \E\bracks[\big]{\psi_i^2(\psi_j^2)^*}
    = \sum_{i=1}^d \E\bracks[\big]{\abs{\psi_i}^4}
    = \frac{1}{d+1}.
\end{align*}
The final inequality follows from a computation using basic facts about the projector on the symmetric subspace over two dimension-$d$ spaces, $\Pi_{\sym} \in \C^{d^2 \times d^2}$~\cite[Proposition 6]{Har13}:
\begin{align*}
    \E\bracks[\big]{\abs{\psi_i}^4}
    = \E\bracks[\big]{\tr(\proj{\psi}^{\otimes 2} \proj{ii})}
    = \frac{1}{d(d+1)}\tr(\Pi_{\sym} \proj{ii})
    = \frac{1}{d(d+1)}.
\end{align*}
Consequently, by Markov's inequality, $\real(\ket{\psi})$ is smaller than $1/10$ with probability $\geq 0.99$, provided $d > 10^3$.

So, consider the following distinguishing algorithm: call reality testing on the input state $\ket{\psi}$, to receive a guess of whether $\real(\ket{\psi}) = 1$ or $\real(\ket{\psi}) < 1/10$ which is correct with probability $\geq 2/3$.
Then, output $\wt{b} = 1$ if the guess is that $\real(\ket{\psi}) = 1$; otherwise, output $\wt{b} = 0$.
When $b = 1$, this algorithm succeeds with probability $\geq 2/3$.
When $b = 0$, it succeeds with probability $\geq 2/3 - 0.01$, since our guarantee on the algorithm's success probability only holds provided $\real(\ket{\psi}) \leq 1/d$.
\end{proof}

\begin{lemma}[{Phase states are indistinguishable from Haar \cite{bs19}}]
    \label{lem:phase-prs}
    Any algorithm which solves \cref{prob:phase-prs} to $> 0.6$ success probability just given copies of the input state must use $\bigOmega{\sqrt{d}}$ copies.
\end{lemma}
\begin{proof}
According to a result of Brakerski and Shmueli~\cite{bs19}, a random phase state $\ket{\psi}$ forms a $\frac{4t^2}{d}$-approximate $t$-design for all $t$.
In particular, using the definition of an approximate $t$-design~\cite[Definition~6]{bs19}, the trace distance between $\sqrt{d} / 5$ copies of a random phase state $\ket{\psi}$ and that many copies of a Haar-random state is $\leq 0.2$, so by the Holevo--Helstrom theorem~\cite[Theorem~3.4]{Wat18}, the probability of success of any distinguishing algorithm is $\leq 0.6$, assuming that $b \in \braces{0,1}$ is chosen uniformly at random.
If an algorithm cannot succeed with probability $> 0.6$ for a randomly chosen $b$, it cannot succeed with that probability for a worst-case $b$.

Thus, any algorithm succeeding at \cref{prob:phase-prs} with success probability $> 0.65$ must use more than $ \sqrt{d} / 5$ copies of the state.
\end{proof}

\begin{corollary}[Reality testing lower bound]
    \label{cor:lower}
    Consider an algorithm which can solve \cref{prob:phase-prs} using $q$ forward and inverse queries to any state preparation unitary $U \in \C^{2d^2 \times 2d^2}$ of $\ket{\psi}$.
    Then the number of queries it uses is $q = \bigOmega{d^{1/4}}$.
\end{corollary}
\begin{proof}
By \cref{lem:phase-prs}, $\bigOmega{\sqrt{d}}$ copies of $\ket{\psi}$ are necessary to solve \cref{prob:phase-prs} with probability $> 0.6$.
Now, suppose we have a circuit which uses $q$ queries to any state preparation unitary of $\ket{\psi}$, and which solves the problem with success probability $\geq 0.65$.
Then by \cref{thm:useless}, this circuit can be converted into a circuit which uses $n = \bigO{q^2}$ copies of $\ket{\psi}$ to do the same, with success probability $\geq 0.64$.
Such an algorithm can only exist provided that $n = \bigOmega{\sqrt{d}}$; in other words, $q = \bigOmega{d^{1/4}}$.
\end{proof}

\begin{proof}[Proof of \cref{thm:useful}]
    This is an immediate consequence of \cref{lem:upper} and \cref{cor:lower}.
\end{proof}

\subsection{Handling transpose queries} \label{subsec:trans}

Our full analysis can similarly be carried out for transpose queries with only minimal changes.
The main change is that we must slightly restrict our definition of a state preparation unitary for the upper bound to go through.

For a pure state $\ket{\psi}$, call a state preparation unitary $U$ for $\ket{\psi}$ \emph{tidy} if it resets its ancillas (up to phase), meaning that, for some $\theta \in [0, 2\pi)$,
\begin{align*}
    U \ket{0}\ket{0} = e^{\ii \theta} \ket{0} \ket{\psi}.
\end{align*}
Then, given a \emph{tidy} state preparation unitary, reality testing can be solved with one query to $U$ and $U^\trans$.
The algorithm is as follows.

\begin{figure}[h]
\begin{equation*}
    \Qcircuit @C=1em @R=.3em @!C @!R {
    \lstick{\ket{0}} & \multigate{1}{U} & \multigate{1}{U^\trans} & \meter \\
    \lstick{\ket{0}} & \ghost{U} & \ghost{U^\trans} & \meter
}
\end{equation*}
\end{figure}
The probability of seeing the outcome $\ket{0}\ket{0}$ is equal to
\begin{align*}
    \abs{\bra{0}\bra{0} U^\trans U \ket{0} \ket{0}}^2
    = \abs{(U^* \ket{0}\ket{0})^\dagger U \ket{0} \ket{0}}^2
    = \abs{\braket{\psi^*}{\psi}}^2
    = \real(\ket{\psi}).
\end{align*}
So, if we see this outcome, we guess that $\real(\ket{\psi}) = 1$; otherwise, we guess that $\real(\ket{\psi}) < 1/10$.
This algorithm succeeds with probability $\geq 2/3$.

In summary, transpose queries can solve the problem of \emph{reality testing given a tidy state preparation unitary}.
To attain a separation, then, we need to show that this task cannot be solved with only $U$ and $U^\dagger$ queries.
The only additional thing to check is that \cref{thm:useless} works even when the state preparation unitaries need to be tidy.
This is true, because the state preparation unitaries that this simulation produces can be made tidy.

Let $V$ be a state preparation unitary such that $\ket{0} \ket{0} = e^{\ii \theta}\ket{1}\ket{\psi}$.
Then the following circuit outputs $e^{\ii \theta} \ket{0}\ket{\psi}$.
\begin{figure}[h]
\begin{equation*}
    \Qcircuit @C=1em @R=.3em @!C @!R {
    \lstick{\ket{0}} & \multigate{1}{V} & \gate{X} & \qw & \qw \\
    \lstick{\ket{0}} & \ghost{V} & \qw & \qw & \qw 
}
\end{equation*}
\end{figure}
So, to simulate a circuit which applies a tidy state preparation unitary $n$ times, we can replace every query with the above gadget.
Then, apply the rest of the algorithm given in \cref{thm:useless}.
This gives the desired simulation.

\section{Conjugates and cryptography}

Here we describe how modeling unitary access without conjugates can lead to incorrect conclusions in the context of cryptography. Namely, we demonstrate a simple commitment scheme relative to a unitary oracle $U$ that is provably secure given query access to $U,U^\dagger$, but insecure if one has access to $U^*$. Since such a commitment would be implemented by a quantum circuit in the ``real world'', an attacker would be able implement $U^*$ for themselves. Hence, in the real world,  the commitment scheme is also insecure.

\subsection{Quantum commitments relative to a unitary oracle}

\newcommand{\As}{{\reg{A}}}
\newcommand{\Cs}{{\reg{C}}}
\newcommand{\Ds}{{\reg{D}}}

We first recall the definition of quantum commitments, as given by Yan~\cite{AC:Yan22}.

\begin{definition}[Commitment syntax]\label{def:commitmentsyntax}
    A canonical (non-interactive) quantum bit commitment is specified by a family of polynomial-time computable quantum states $\{|\psi_{\lambda,b}\rangle\}_{\lambda,b\in\{0,1\}}$ over a pair of registers $\Cs_\lambda$ and $\Ds_\lambda$. It consists of two stages:
    \begin{itemize}
        \item \textbf{Commit.} To commit to a bit $b$, the sender prepares the state $|\psi_{\lambda,b}\rangle_{\Ds_\lambda,\Cs_\lambda}$, and sends register $\Cs_\lambda$ to the receiver as the commitment register and keeps $\Ds_\lambda$ as the opening register.
        \item \textbf{Open.} To open the commitment, the sender sends register $\Ds_\lambda$ and the bit $b$ to the receiver. The receiver projects the state contained in $(\Ds_\lambda,\Cs_\lambda)$ onto $|\psi_{\lambda,b}\rangle\langle\psi_{\lambda,b}|$. If the projection accepts, it outputs $b$; otherwise it outputs $\bot$.
    \end{itemize}
\end{definition}

Quantum commitments also must satisfy hiding and binding. 

\begin{definition}[Perfect hiding] A quantum commitment is perfectly hiding if we have the following equality of reduced density matrices: \[\left(|\psi_{\lambda,0}\rangle\langle\psi_{\lambda,0}|\right)_{\Cs_\lambda}=\left(|\psi_{\lambda,1}\rangle\langle\psi_{\lambda,1}|\right)_{\Cs_\lambda}\]
\end{definition}
In other words, commitments to 0 and 1 look identical to the receiver.

For binding, we consider honest-binding, which informally guarantees that the sender who honestly computes the commitment cannot later change the committed bit. It is known that honest-binding for canonical form commitments implies the stronger notion of security where the adversary may not even have honestly generated the commitment~\cite{AC:Yan22}.

\begin{definition}[Honest computational binding]\label{def:honestbinding}
    A commitment scheme $$\{|\psi_{\lambda,b}\rangle\}_{\lambda,b}$$ is computationally \emph{honest-binding} if, for every quantum polynomial time adversary $A$ acting on register $\As_\lambda,\Ds_\lambda$ ($\As_\lambda$ being private registers for the adversary), and for every bit $b\in\{0,1\}$, there exists a negligible function $\epsilon$ such that the following holds. Let $\rho_b$ be the reduced density matrix obtained by applying $A$ to the $(\As_\lambda,\Ds_\lambda)$ registers of the state $|0\rangle_{\As_\lambda}|\psi_{\lambda,b}\rangle_{\Ds_\lambda,\Cs_\lambda}$ and then tracing out $\As_\lambda$. Then
    \begin{equation}\label{eq:commitment}
        \tr\parens[\big]{\;|\psi_{\lambda,1-b}\rangle\langle\psi_{\lambda,1-b}|\rho_b\;}\leq\epsilon
    \end{equation}
    $\epsilon$ is called the ``advantage'' of $A$.
\end{definition}
In other words, $\As$ cannot transform a commitment to $b$ to a commitment to $1-b$ by just acting on the register $\Ds_\lambda$.

\paragraph{Commitments relative to oracles.} We can define quantum commitments relative to unitary oracles $(U^1_\lambda,\cdots, U^k_\lambda)$ as a straightforward modification of the above definitions. The states $|\psi_{\lambda,b}\rangle$ are now specified relative to $(U^1_\lambda,\cdots, U^k_\lambda)$, denoted $|\psi^{U^1_\lambda,\cdots,U^k_\lambda}_{\lambda,b}\rangle$, and the algorithm computing these states is allowed to make queries to $(U^1,\cdots,U^k)$. Likewise, the receiver will also be allowed to make queries to $(U^1,\cdots,U^k)$. Moreover, the adversary $A_\lambda$ now can make queries to $(U^1,\cdots,U^k)$. The restriction to being polynomial time computable implies that the number of queries made to construct the commitment states and also by the adversary is polynomial. Hiding will require identical reduced density matrices for every possible oracle. Binding will be defined relative to a distribution over oracles, where the trace taken in Equation~\ref{eq:commitment} is averaged over the choice of $U^1,\cdots,U^k$.

\subsection{A simple commitment scheme}

We now define a simple commitment scheme as follows. Let $\{U_\lambda\}_\lambda$ be a family of unitaries acting on register $\Cs_\lambda=\Ds_\lambda$. The states $|\psi^{U_\lambda}_{\lambda,b}\rangle$ are defined as follows:
\begin{align*}
    |\psi^{U_\lambda}_{\lambda,0}\rangle&=|\epr\rangle_{\Ds_\lambda,\Cs_\lambda}:=\sum_x |x\rangle_{\Ds_\lambda}|x\rangle_{\Cs_\lambda}\\
    |\psi^{U_\lambda}_{\lambda,1}\rangle&=(\mathbf{I}_{\Ds_\lambda}\otimes U_\lambda)\epr\rangle
\end{align*}

Note that the projection onto $|\psi^{U_\lambda}_{\lambda,b}\rangle$ requires access to $U_\lambda$ and $U^\dagger_\lambda$. So the commitment scheme satisfies the requirements of Definition~\ref{def:commitmentsyntax} as long as the oracles given out include $U_\lambda$ and $U_\lambda^\dagger$.

\subsection{Hiding}

\begin{lemma}The commitment scheme $\{|\psi_{\lambda,b}\rangle\}_{\lambda,b}$ is perfectly hiding.\end{lemma}
\begin{proof}Observe that \begin{equation}\label{eq:commitmentconj}|\psi^{U_\lambda}_{\lambda,1}\rangle=(\mathbf{I}_{\Ds_\lambda}\otimes U_\lambda)|\epr\rangle=(U_\lambda^*\otimes\mathbf{I}_{\Cs_\lambda})|\epr\rangle\end{equation}
This means $|\psi^{U_\lambda}_{\lambda,0}\rangle$ and $|\psi^{U_\lambda}_{\lambda,1}\rangle$ differ by a unitary applied to the register $\Ds_\lambda$. As such, tracing out the $\Ds_\lambda$ register on both states results in the same reduced density matrix over $\Cs_\lambda$.\end{proof}

\subsection{Insecurity of binding with \texorpdfstring{$U^*$}{conjugate queries}}

\begin{lemma}The commitment scheme $\{|\psi_{\lambda,b}\rangle\}_{\lambda,b}$ is not binding if the adversary is given the oracle $U_\lambda^*$, or alternatively both $U_\lambda$ and $U_\lambda^\trans$.\end{lemma}
\begin{proof}The adversary commits to $0$ by sending half of $|\psi_{\lambda,0}^{U_\lambda}\rangle=|\epr\rangle$.
Equation~\ref{eq:commitmentconj} also shows that a binding adversary can transform $|\psi^{U_\lambda}_{\lambda,0}\rangle$ to $|\psi^{U_\lambda}_{\lambda,1}\rangle$, the commitment of 1, by simply applying $U_\lambda^*$ to $\Ds_\lambda$. The overall attack (including generating the commitment) requires only a single query to $U^*_\lambda$.

We can also change from commitments to 1 to commitments to 0 by applying $U_\lambda^\trans$. Note that a commitment to 1 requires a single query to $U_\lambda$ to construct, so the overall attack requires a query to both $U_\lambda$ and $U_\lambda^\trans$.
\end{proof}

\subsection{Security with just \texorpdfstring{$U,U^\dagger$}{forward and inverse queries}}

We now show, however, that with access to $U_\lambda,U_\lambda^\dagger$ but without direct access to $U_\lambda^*$, it is impossible in polynomial time to break binding. In particular, this means that the commitment scheme is valid and secure in a model where all parties are given $(U_\lambda,U_\lambda^\dagger)$ (since the construction only needs $U_\lambda$ and $U_\lambda^\dagger$). However, such binding security would generally be considered incorrect, as in ``the real world'' one would have a circuit for $U_\lambda$, which allows for computing $U_\lambda,U_\lambda^\dagger$ but also $U_\lambda^*$ and even $U_\lambda^\trans$. 

\begin{theorem}\label{thm:commitmentbinding} For Haar random $U_\lambda$, the commitment scheme $\{|\psi^{U_\lambda}_{\lambda,b}\rangle\}_{\lambda,b}$ is honest computational binding if the adversary is given only $U_\lambda,U_\lambda^\dagger$.
\end{theorem}
We will prove Lemma~\ref{thm:commitmentbinding} by gradually reducing the problem to ever-simpler tasks, ultimately arriving at the task of constructing $|\phi^*\rangle$ from several copies of $|\phi\rangle$ for a Haar-random state $|\phi\rangle$, which is information-theoretically impossible. From this point forward, we will ignore computational costs and just consider query access. Query lower-bounds in particular imply computational lower-bounds. Throughout, we will focus on the security of transforming commitments to 0 to commitments to 1, the other direction having an almost identical proof.

\paragraph{Random $U_\lambda$ to arbitrary $U_\lambda$.} We now show that if there existed an adversary contradicting the security of $\{|\psi_{\lambda,b}\rangle\}_{\lambda,b}$ in the $U_\lambda,U_\lambda^\dagger$ model for a Haar-random $U_\lambda$, then there is an adversary for \emph{any} distribution over $U_\lambda$.
\begin{lemma}If there exists a quantum algorithm $A_0$ making $q$ queries, which has advantage $\epsilon$ in breaking honest binding for $\{|\psi^{U_\lambda}_{\lambda,b}\rangle\}_{\lambda,b}$ for a Haar-random $U_\lambda$, then for any other distribution $D_\lambda$ on $U_\lambda$, there exists a quantum algorithm $A_1$ making $q$ queries, which also has advantage $\epsilon$ in breaking honest binding.
\end{lemma}
\begin{proof}We devise an adversary $A_1$ for the distribution $D_\lambda$ as follows. $A_1$ makes queries to $V_\lambda,V_\lambda^\dagger$ sampled from $D_\lambda$, and attempts to map $|\psi^{V_\lambda}_{\lambda,b}\rangle$ to $|\psi^{V_\lambda}_{\lambda,1-b}\rangle$.

To do so, it will sample a Haar-random unitary $W_\lambda$, and define $U_\lambda=W_\lambda V_\lambda$. It will run $A_0$ with oracles $U_\lambda,U_\lambda^\dagger$, which it can simulate using $W_\lambda$ and by making queries to $V_\lambda$. Observe that $U_\lambda$ is in fact Haar random.

The guarantee of $A_0$ is that on input $|\psi^{U_\lambda}_{\lambda,0}\rangle$, it produces a state with $\epsilon$-overlap with \[|\psi^{U_\lambda}_{\lambda,1}\rangle=(U_\lambda^*\otimes\mathbf{I})|\epr\rangle=(W_\lambda^*\otimes\mathbf{I})|\psi^{V_\lambda}_{\lambda,1}\rangle\]
just by manipulating the $\Ds_\lambda$ register. Then $A_1$ can just apply $W^\trans$ to the $\Ds_\lambda$ register to get a state with $\epsilon$ overlap with $|\psi^{V_\lambda}_{\lambda,1}\rangle$.
\end{proof}
\begin{remark}Here, generating $W_\lambda$ and running it is in general inefficient, which is okay for us since we are focusing on query complexity. However, $W_\lambda$ can be made efficient by replacing it with a PRU.
\end{remark}

\paragraph{Our distribution of unitaries.} Moving forward, we will choose $D_\lambda$ to be the distribution which samples a random subspace of half the dimension, and sets $U_\lambda$ to be the reflection about that subspace. Equivalently, $D_\lambda$ randomly partitions the space into the product of random subspaces $S_i$ of dimension 2, then within each subspace chooses a random basis $\{|\phi_{i,0}\rangle,|\phi_{i,1}\rangle\}$ and has $U_\lambda$ swap $|\psi_{i,0}\rangle$ and $|\psi_{i,1}\rangle$. Then $U_\lambda$ is reflecting about the space spanned by $|\phi_{i,0}\rangle+|\phi_{i,1}\rangle$. Let $B$ be the basis for the entire space $\{|\phi_{i,b}\rangle\}_{i,b}$.

\paragraph{Swapping elements of $B$.} We now show that $A_1$ must be swapping elements of $B$. Consider the following different experiment on the adversary $A_1$: a random computational basis element $|\phi_{i,b}\rangle$ is chosen and fed into $A$. Then let $\tau_{i,b}^{U_\lambda}$ be the reduced density matrix after applying $A_1$ and tracing out the register $\As_\lambda$. Then we let $p=\min_{i,b}\mathbb{E}_{U_\lambda\gets D_\lambda}[\langle \phi_{i,1-b}^*|\; \tau_{i,b}^{U_\lambda}\;|\phi^*_{i,1-b}\rangle]$
We call $p$ the basis state swap advantage of $A$.
\begin{lemma}For any $A_1$, $\epsilon\leq p$.
\end{lemma}
\begin{proof}Let $\Pi$ be the projection onto the states $\left(|\phi^*_{i,1-b}\rangle_{\Ds_\lambda}\right)|\phi_{i,b}\rangle_{\Cs_\lambda}=\left(U^*|\phi^*_{i,b}\rangle_{\Ds_\lambda}\right)|\phi_{i,b}\rangle_{\Cs_\lambda}$. Let $p'=\mathbb{E}_{i,b}\mathbb{E}_{U_\lambda\gets D_\lambda}[\langle \phi_{i,1-b}^*|\; \tau_{i,b}^{U_\lambda}\;|\phi^*_{i,1-b}\rangle]$. Then observe that $p'$ is equivalent to \[p'=\tr(\Pi\rho_0)\]
where $\rho_0$ is the reduced density matrix over $(\Ds_\lambda,\Cs_\lambda)$ as defined in Definition~\ref{def:honestbinding}. The space accepted by $\Pi$ includes in particular the state $|\psi_{\lambda,1}\rangle$, meaning $p'=\tr(\Pi\rho_0)\geq \tr(|\psi_{\lambda,1}\rangle\langle\psi_{\lambda,1}|\rho_0)=\epsilon$.

To then lower-bound $p$, observe that the choice of indices $i,b$ in the basis states is completely opaque to $A_1$, and as such the average is equal to the worst case choice of $i,b$.
\end{proof}

In other words, the algorithm $A_1$ is good at synthesizing the states $|\phi_{i,1-b}^*\rangle$ given $|\phi_{i,b}^*\rangle$, for an arbitrary choice of $i,b$.

\paragraph{Moving to a simple swap oracle.} We now remove all of the oracle $U_\lambda$ except the part that acts on the subspace $S_0$.

Let $D_\lambda'$ be the distribution which samples a random subspace $S$ of dimension 2, a random basis $|\phi_0\rangle,|\phi_1\rangle$ for that basis, and produces the unitary $V_\lambda$ which swaps $|\phi_0\rangle$ with $|\phi_1\rangle$ but acts as the identity on states orthogonal to $S$. Equivalently, $V_\lambda$ is the reflection $\mathbf{I}-(|\phi_0\rangle-|\phi_1\rangle)(\langle\phi_0|-\langle\phi_1|)/2$.

\begin{lemma}\label{lem:movetoswap}Let $A_1$ and $q,p,\epsilon$ be as above. Then there exists a algorithm $A_2$ which is given as input $|\phi_0\rangle$ and makes $q$ queries to $V_\lambda$ and $\bigO{q\times t}$ queries to the projector for $S$, where $S,|\phi_0\rangle,|\phi_1\rangle$ and $V_\lambda$ are sampled from $D_\lambda'$. The guarantee on $A_2$ is the following. Let $\upsilon$ be the reduced density matrix of $A_2$ after tracing out $\As_\lambda$. Then $\langle\phi_1^*| \upsilon |\phi_1^*\rangle\geq \epsilon-2^{-\bigO{t}}-\bigO{1/d}$, where $d$ is the dimension.
\end{lemma}
\begin{proof}$A_2$ will construct a subspace $T$ which is orthogonal to $S$ and has dimension $d/2-1$, where $d$ is the total dimension. We will explain later how $T$ is constructed. Then $A_2$ will define the action of $U_\lambda$ as follows:
\begin{itemize}
    \item On $S$, $U_\lambda$ will be equal to $V_\lambda$.
    \item On $T$, $U_\lambda$ will be a sign flip
    \item Orthogonal to $S,T$, $U_\lambda$ will be the identity.
\end{itemize}
This fully specifies $U_\lambda$. Moreover, if $T$ is a random subspace orthogonal to $S$, then $U_\lambda$ is exactly reflecting about a space of dimension $d/2$ ($d/2-1$ from the space orthogonal to $S,T$, and $1$ from the subspace of $S$ that is not flipped). $A_2$ then runs $A_1$ with access to this $U_\lambda$, where the input state to $A_1$ is the same as $A_2$'s input state. This implicitly sets the subspace $S_i$ to be $S$ and the state $|\phi_{i,b}\rangle=|\phi_0\rangle$, for an arbitrary $i,b$.

So far, our description of $A_2$ has success probability matching $A_1$, and in particular at least $\epsilon$. We just need to describe how to implement $U_\lambda$. Notice that applying $V_\lambda$ takes care of the action on the subspace $S$ and leaves everything else untouched. It then suffices to perform a sign flip on $T$, which is equivalent to being able to project onto $T$.

The challenge is that we do not know $S$ explicitly, but only have oracle access, so it is not feasible for us to exactly construct a subspace orthogonal to $S$. However, we can implicitly construct a subspace $T$, together with the ability to project onto it, as follows:
\begin{itemize}
    \item We choose a random subspace $T'$ of the entire space, with dimension $d/2-1$.
    \item We implicitly define $T$ as the space orthogonal to both $S$ and $T'$. With all but negligible probability, $T$ will have dimension $d/2-1$. By definition, $T$ will be orthogonal to $S$, and it will be random conditioned on orthogonality with $S$.
    \item In order to (approximately) project onto $T'$, we iteratively apply the projectors for $T'$ and $S$ for $t$ steps, and accept if all the projectors rejected. 
\end{itemize}
The above clearly accepts states in $T$. We just need to show that it rejects anything orthogonal to $T$ with overwhelming probability. To do so, consider the matrix $M=(I-T')(I-S)(I-T')$ where we take $S,T'$ to also be the projectors onto the subspaces $S,T'$. Let $|\tau\rangle$ be an eigenvector of $M$ with eigenvalue $\lambda$. Then the probability iteratively applying $T',S$ to $|\tau\rangle$ for $t$ steps will have all projectors reject with probability 
\begin{align*}
    \| ST'ST'\cdots ST'|\tau\rangle\|^2&=\langle\tau|(I-T')(I-S)\cdots (I-T')(I-S)\enspace (I-S) (I-T')\cdots (I-S)(I-T')|\tau\rangle\\
    &=\langle\tau|M^{2t-1}|\tau\rangle=\lambda^{2t-1}
\end{align*}
A state is in $T$ if and only if it is an eigenvector of $M$ with eigenvalue $\lambda=1$, in which case we get acceptance with probability 1, as expected. Let $\lambda_0$ be the maximal eigenvalue of $M$ that is less than 1, which must be for a vector orthogonal to $T$. For general states orthogonal to $T$, the maximal acceptance probability is obtained by this maximal eigenvalue, so is $\lambda_0^{2t-1}$.

We claim that $\lambda_0$ is almost certainly a constant. Indeed, it is not hard to show that there are exactly two eigenvalues of $M$ that are between 0 and 1 (corresponding to the projection of $S$ onto the $I-T'$), and that they concentrate around 1/2. \end{proof}

\paragraph{Distinguishing $|\phi_1^*\rangle$ random.} Now we take the algorithm $A_2$ guaranteed in Lemma~\ref{lem:movetoswap}, and use it to construct a distinguisher which distinguishes copies of $|\phi_1\rangle$ and $|\phi_1^*\rangle$ from copies of two independent Haar-random states.

\begin{lemma}\label{lem:movetodistinguish}Let $A_2$ be as above. For a state $|\phi\rangle$ orthogonal to $|0\rangle$, let $\mathrm{SWAP}_{|\phi\rangle}$ be the unitary which swaps $|0\rangle$ and $|\phi\rangle$, leaving any state orthogonal to $|0\rangle,|\phi\rangle$ untouched.

Then there exists an algorithm $A_3$ making $\bigO{qt}$ queries to $\mathrm{SWAP}_{|\phi_1\rangle}$ and 1 query to either $\mathrm{SWAP}_{|\phi_1^*\rangle}$ or $\mathrm{SWAP}_{|\phi'\rangle}$, where $|\phi'\rangle$ is an independent Haar random state orthogonal to $|0\rangle$. $A_3$ distinguishes the two cases with probability at least $\epsilon-2^{-\bigO{t}}-\bigO{1/d}$.
\end{lemma}
Note that the $\mathrm{SWAP}$ oracle in Lemma~\ref{lem:movetodistinguish} (which swaps two states) is different from the procedure $\mathrm{SWAP}$ used in Definition~\ref{def:shiftgate} (which swaps two registers).
\begin{proof}We first note that the projector onto $S$ can be implemented by projectors onto $|\phi_0\rangle$ and $|\phi_1\rangle$. We also note that we can take the description of $|\phi_0\rangle$ to be known, and this will only improve the success probability. In this case, we might as well call $|\phi_0\rangle$ the state ``$|0\rangle$''. $V_\lambda$ then swaps $|0\rangle$ with $|\phi_1\rangle$ (which is Haar random and orthogonal to $|0\rangle$). The projector onto $|\phi_1\rangle$ is easily implemented with queries to $V_\lambda$.

$A_3$ will then run $A_2$ on input $|0\rangle=|\phi_0^*\rangle$ and using oracle $V_\lambda=\mathrm{SWAP}_{|\phi_1\rangle}$, and get a state that is purportedly $|\phi_1^*\rangle$. It will run this through the second oracle, and check if the result is $|0\rangle$. If the second oracle swaps with $|\phi_1^*\rangle$, then this test will pass with probability at least $\epsilon-2^{-\bigO{t}}$. If on the other hand it passes with an independent $|\phi'\rangle$, the probability of passing is at most $\bigO{1/d}$.
\end{proof}

\paragraph{Simulating the swap oracles.} We now use Lemma~\ref{lem:reflect} to simulate the oracles $\mathrm{SWAP}$. For a state $|\phi\rangle$ that is not supported on $|0\rangle$, let $|\phi-\rangle=(|0\rangle-|\phi\rangle)/\sqrt{2}$.

\begin{lemma}Let $A_3$ be as above. Then there is an algorithm $A_4$ which is given $|\phi_1-\rangle^{\otimes\ell} |\phi'-\rangle^{\otimes \ell}$ where $|\phi_1\rangle$ is Haar random and $|\phi'\rangle$ is either an independent Haar random state or $|\phi^*_1\rangle$. $A_4$ distinguishes the two cases with probability at least $\epsilon-2^{-\bigO{t}}-\bigO{1/d}-\bigO{qt/\ell}$. 
\end{lemma}
\begin{proof}We observe that $\mathrm{SWAP}_{|\phi\rangle}$ just reflects $|\phi-\rangle$. Therefore, we apply Lemma~\ref{lem:reflect} using our copies of $|\phi_1-\rangle^{\otimes\ell} |\phi'-\rangle^{\otimes \ell}$ to simulate the oracles $\mathrm{SWAP}_{|\phi_1\rangle}$ and $\mathrm{SWAP}_{|\phi'\rangle}$ used by $A_3$. Each query incurs an error $\bigO{1/\ell}$, leading to an overall error of $\bigO{qt/\ell}$.
\end{proof}

\paragraph{Distinguishing is impossible.} Finally, we show that the derived algorithm $A_4$ is impossible unless $\epsilon$ is very small. For this, we use that for Haar random states, it is impossible to distinguish conjugates from independent Haar random states (e.g.~\cite{zhandry25}, Lemma 21). The slight variant is that we use the $|\phi-\rangle$ states instead of the original Haar random states themselves. Nevertheless, a straightforward adaptation of the result for Haar-random states applies to the $|\phi-\rangle$ states as well:

\begin{lemma}[Slight variant of \cite{zhandry25}, Lemma 21] Let $|\phi_1\rangle$ and $|\phi'\rangle$ denote independent Haar random states with support outside $|0\rangle$. Let $d$ be the total dimension. Then \[\left\|\E\left[\left(|\phi_1-\rangle\langle\phi_1-|\right)^{\otimes \ell}\otimes \left(|\phi_1^*-\rangle\langle\phi_1^*-|\right)^{\otimes \ell}\right]-\E\left[\left(|\phi_1-\rangle\langle\phi_1-|\right)^{\otimes \ell}\otimes \left(|\phi'-\rangle\langle\phi'-|\right)^{\otimes \ell}\right]\right\|\leq \bigO[\big]{\frac{\ell^2}{d}}\].
\end{lemma}

We are now ready to finish the proof of Theorem~\ref{thm:commitmentbinding}. By combining all the lemmas, we obtain that $\epsilon\leq \bigO{\ell^2/d+qt/\ell+2^{-\bigO{t}}}$. We can now tune the parameters $\ell,t$ to get a good upper bound on $\epsilon$ in terms of $q,d$.  See $t=\bigO{\log d}$ so that $2^{-\bigO{t}}=1/d$, which gives $\epsilon\leq \bigO{\ell^2/d+q\log(d)/\ell}$. Then set $\ell=(q d\log(d))^{1/3}$, which gives $\epsilon\leq \bigO{q^{2/3} \log(d)^{2/3}/d^{1/3}}$. If we take $d$ to be exponential in $\lambda$ (corresponding to a linear number of qubits), then $\epsilon$ is negligible for any polynomial $q$, as desired. This completes the proof of Theorem~\ref{thm:commitmentbinding}.

\section*{Acknowledgments}

We thank Fermi Ma and Robin Kothari for enlightening discussions.
We thank Gregory Rosenthal for carefully reading an earlier version of this document.
E.T.\ is supported by the Miller Institute for Basic Research in Science, University of California Berkeley.
J.W.\ is supported by the NSF CAREER award CCF-233971.

\printbibliography

\end{document}